\documentclass[3p]{elsarticle}
\usepackage[titletoc,toc,title]{appendix}
\usepackage{euscript,amsmath,amssymb,amsfonts,graphicx,bm}
\usepackage{hyperref}
\hypersetup{colorlinks=true,linkcolor=blue,citecolor=red}
\usepackage{epstopdf}
\usepackage{enumerate}
\usepackage{csquotes}
\usepackage{caption}
\usepackage{subcaption}
\usepackage{mathrsfs}
\usepackage{xcolor}
\usepackage{tikz}
\usetikzlibrary{datavisualization.formats.functions}

\numberwithin{equation}{section}

\usepackage{amsmath}
\usepackage{algorithm}
\usepackage[noend]{algpseudocode}

\makeatletter
\def\BState{\State\hskip-\ALG@thistlm}
\makeatother

\newtheorem{thm}{Theorem}[section]

\newtheorem{lem}{Lemma}[section]

 \newtheorem{ass}{Assumption}
\newtheorem{proof}{Proof}[section]

\newcommand{\x}{\mathbf{x}}

\newcommand{\reals}{\mathbb{R}}

\newcommand{\integers}{\mathbb{Z}}
\newcommand{\naturals}{\mathbb{N}}
\newcommand{\prob}{\mathbb{P}}
\newcommand{\var}{{\rm Var }}

\newcommand{\expect}{\mathbb{E}}

\newcommand{\n}{\mathbf{n}}

\newcommand{\bfone}{\mathbf{1}}

\newcommand{\mcF}{\mathcal{F}}

\newcommand{\mcR}{\mathcal{R}}

\newcommand{\mcJ}{\mathcal{J}}

\newcommand{\mcX}{\mathcal{X}}
\newcommand{\mcY}{\mathcal{Y}}

\newcommand{\sfC}{\mathsf{C}}

\begin{document}

\begin{frontmatter}
\title{Coupling sample paths to the partial thermodynamic limit in stochastic chemical reaction networks}

\author{Ethan Levien$^1$}
\ead{levien@math.utah.edu}

\author{Paul C. Bressloff$^{1}$}
\ead{bressloff@math.utah.edu}

\address{$^1$Department of Mathematics, University of Utah, Salt Lake City, UT 84112 USA}

\begin{abstract}
Many biochemical systems appearing in applications have a multiscale structure so that they converge to piecewise deterministic Markov processes in a thermodynamic limit.  The statistics of the piecewise deterministic process can be obtained much more efficiently than those of the exact process. We explore the possibility of coupling sample paths of the exact model to the piecewise deterministic process in order to reduce the variance of their difference. We then apply this coupling to reduce the computational complexity of a Monte Carlo estimator. In addition to rigorous results concerning the asymptotic computational complexity of the Monte Carlo estimator, numerical simulations are performed on some simple biological models suggesting that significant computational gains are made. 
\end{abstract}


\end{frontmatter}


\section{Introduction}

Large stochastic biochemical reaction networks are a popular modeling framework for investigating cellular processes \cite{bressloff2014book}, but the complexity and population sizes involved in realistic models pose major computational challenges. However, when there is a separation of scales, such models lend themselves to a number of model reduction techniques that are useful for course grained analysis.  One example occurs when there is a separation in species abundances \cite{crudu2009,eve2007}. If some subset of chemical species in a reaction network are extremely abundant, then reaction channels involving those species will generally occur much faster than reactions involving less abundant species. In this case, one can take a partial thermodynamic limit to obtain a piecewise deterministic Markov process (PDMP). A number of recent studies have provided rigorous errors bounds for this type of reduction \cite{jahnke2012,crudu2009,Chevallier2016,kurtz1972}. While the PDMP yields useful information about stochastic effects of the rare species, quantitative information about the stochastic fluctuations of the abundant species is lost. On the other hand, in many systems, particularly those with feedback between the rare and abundant chemical species, there is an interest in quantifying the stochastic effects due to these fluctuations \cite{Anderson2015book}. A common method for resolving these fluctuations is the diffusion approximation. While the diffusion approximation is often thought to be computationally advantageous, recent work on classically scaled population models has shown that this methods yields only moderate computational gains \cite{anderson2015b}. Moreover, the error between the PDMP and the exact model is fixed and it is sometimes desirable to control this quantity, especially when the separation of scales is only moderate.  

An alternative to multiscale reduction techniques is to develop methods for accelerating stochastic simulation algorithms such as the Gillespie algorithm \cite{Gillespie77,Gillespie07,Zeiser08}. For example, there have been numerous studies on the method of $\tau$-leaping in an effort to accelerate simulations of continuous time Markov chains \cite{Cao06,anderson2011}. More recently, multi-level methods that couple tau-leaping approximations at different resolutions have been used to reduce variances in Monte Carlo estimators \cite{anderson2014,anderson2012,anderson2015b}. Variance reduction techniques that utilize probabilistic couplings have also appeared earlier in the context of SDEs and Markov Chain Monte Carlo methods \cite{goodman2009,giles2015}. While there has been some work that leverages multiscale reduction techniques for Monte Carlo estimators \cite{Ganguly2014}, to our knowledge the idea of using these techniques directly as a variance reduction tool has not been studied.

In this paper we explore the idea of coupling reduced models to exact models as a variance reduction tool for Monte Carlo estimators.  We develop a new efficient Monte Carlo estimator for multiscale chemical reaction networks near a partial thermodynamic limit. The key insight is that, since only a small fraction of the degrees of freedom of the PDMP are stochastic, one can efficiently compute statistics of the process without Monte Carlo simulations using non Monte Carlo based techniques. On the other hand, if one wants to resolve demographic noise in the full model it is necessary to perform a large number of Monte Carlo simulations. By coupling the full stochastic model to the PDMP one can reduce the variance by a factor inversely related to the system size, and hence a smaller number of simulations need to be performed to achieve a given error tolerance. For practical applications the desired error tolerance of the Monte Carlo estimator scales with this factor. Hence the coupled Monte Carlo estimator has the potential to speed up computations by a fractional power of the error tolerance. Our results extend the idea of variance reduction developed in \cite{anderson2014,anderson2012} and provide a new computational application of the theory developed in previous work on PDMP approximations, or partial thermodynamic limits \cite{jahnke2012,crudu2009,Chevallier2016}.
\par
The paper is organized as follows. In Section \ref{sec_basics} we give a introduction to stochastic chemical reaction networks in the single scale, or classical setting. This includes an introduction to the random time change representation of Kurtz, the Gillespie algorithm and Monte Carlo simulations. In Section \ref{sec_ptdl} we introduce multiscale models and show how to obtain a piecewise deterministic Markov process by taking a partial thermodynamic limit. Section \ref{sec_mce} introduces a framework for variance reduction in Monte Carlo estimators using the theory in Section \ref{sec_ptdl}. Finally, Section  \ref{sec_coupling} and  \ref{sec_couplingfb} give an analysis of the computational complexity of the coupled Monte Carlo estimator, while numerical examples are presented in Section \ref{sec:numerical}. Non-Monte Carlo based methods for obtaining statistics of the reduced model are presented in Appendix \ref{sec_nonmce}, while our technical results are proved in Appendix \ref{ap:proof1} and \ref{ap:proof2}.

\section{Representation and simulation of stochastic chemical reaction networks in the classical setting}\label{sec_basics}

In this section we briefly introduce some background material pertaining to stochastic chemical reaction networks in the classical setting. Numerous books and articles provide a more comprehensive review of this material \cite{Anderson2015book,ethier86}. A chemical reaction network provides a mathematical description a set of interacting species, denoted $\mcX = \{\mcX_i,\, i=1,\ldots,d\}$. In general $d$ may be infinite and although our formulation is completely general, we are mostly concerned with examples where $d$ is finite. These interactions are defined by a set of single-step reactions $\mcR = \{\mcR_j,\, j=1,\ldots,p\}$. Let $x_i$ be the number of molecules of $\mcX _i$ and set $x=(x_1,\ldots,x_d)$. The $j$-th reaction takes the form
\[{\mathcal R}_j:\, \sum_{i=1}^dk_{j,i}^{\rm in}\mcX_i\rightarrow \sum_{i=1}^dk_{j,i}^{\rm out}\mcX_i,\]
where $k_{j,i}^{\rm in},k_{j,i}^{\rm out}$ are known as {\em stochiometric coefficients}. When such a reaction occurs the state $\n$ is changed according to
\[x_i\rightarrow x_i+k_{j,i},\quad k_{j,i}=k_{j,i}^{\rm out}-k_{j,i}^{\rm in}.\]
More complicated multi-step reactions can always be decomposed into these fundamental single-step reactions with appropriate stochiometric coefficients. In practice, most reactions involve collisions between pairs of molecules, so that we take $\sum_ik_{j,i}^{\rm in}=1$ or $2$. 

In the case of a large number of molecules for each species, one often describes the dynamics of a reaction network in terms of a deterministic kinetic or rate equation involving the concentrations $z_j=x_j/\Omega$ -- the {\em law of mass action} \cite{bressloff2014book}. Here $\Omega$ is a dimensionless quantity representing the system size, which in gene networks is typically taken to be the characteristic number of proteins. Alternatively, it could represent some volume scale factor. Currently, we are in the \emph{classical setting} where $\Omega$ scales the abundances of every species. This is in contract to the \emph{multiscale setting} where only the size of some subsystem is scaled with the system size.   \par  For a set of $p$ reactions, the kinetic equations take the form (strictly speaking in the thermodynamic limit $\Omega \rightarrow \infty$)
\begin{equation}
\label{ostoch}
\frac{dz_i}{dt}=\sum_{j=1}^p k_{j,i}\left [\zeta_j \prod_{i=1}^dz_i^{k^{\rm in}_{j,i}}\right ]\equiv \sum_{j=1}^pk_{j,i} \alpha_j(z),
\end{equation}
where $\zeta_j$ is a constant that depends on the probability that a collision of the relevant molecules actually leads to a reaction, and $z=x/\Omega$. The product of concentrations is motivated by the idea that in a well-mixed container there is a spatially uniform distribution of each type of molecule, and the probability of a collision depends on the probability that each of the reactants is in the same local region of space. Ignoring any statistical correlations, the latter is given by the product of the individual concentrations. The functions $\alpha_j$ are known as transition intensities or {\em propensities}. Given this notation, it is straightforward to write down the corresponding chemical master equation for finite $\Omega$, which takes into account intrinsic fluctuations in the number of molecules (demographic noise). Setting $P(x,t)=\prob(x(t)=x|x(0)=x_0)$, the chemical master equation is
\begin{equation}
\label{masterg}
\frac{dP(x,t)}{dt}=\Omega \sum_{j=1}^p\left (\prod_{i=1}^d{\mathbb E}^{-k_{j,i}}-1\right )\alpha_j(x/\Omega)P(x,t),
\end{equation}
Here ${\mathbb E}^{-k_{j,i}}$ is a step or ladder operator such that for any function $g(x)$,
\begin{equation}
{\mathbb E}^{-k_{j,i}}g(x_1,\ldots,x_i,\ldots,x_d)=g(x_1,\ldots, x_i-k_{j,i},\ldots,x_d).
\end{equation}
One point to note is that when the number of molecules is sufficiently small, the characteristic form of a propensity function $\alpha(\x)$ in equation (\ref{ostoch}) has to be modified:
\[ \left (\frac{x_i}{\Omega}\right )^{k^{\rm in}_{j,i} }\rightarrow \frac{1}{\Omega^{k^{\rm in}_{j,i}}}\frac{x_i!}{(x_i-k^{\rm in}_{j,i})!}.\]

In this paper, it will be more convenient to consider an alternative representation of a chemical reaction network. If we let the random variable $X_i(t)$ denote the number of molecules in species $i$, then we can express $X_i(t)$ as
\begin{equation*}
X_i(t) = X_i(0) + \sum_{j}R_j(t)k_{j,i}
\end{equation*}
where $R_j(t)$ denotes the number of times reaction $j$ has occurred by time $t$. For example, if the propensity function $\alpha_j$ is constant for each $j$, $R_j(t)$ is a Poisson process with rate $\Omega \alpha_j$, so that letting $\Pi_j$ denote a unit rate Poisson process gives us
\begin{equation*}
X_i(t) = X_i(0) +  \sum_{j}\Pi_j(\Omega \alpha_j t)k_{j,i}.
\end{equation*}
More generally, if $\alpha_j$ depends on $X(t) = \{X_i(t)\}_{i=1,\ldots,d}$ we obtain
\begin{equation}\label{sde}
X_i(t) = X_i(0) + \sum_{j}\Pi_j\left(\Omega \int_0^t\alpha_j(X(s)/\Omega)ds\right)k_{j,i}.
\end{equation}
This representation is due to Kurtz and is known as the random time change representation of the process $X(t)$ \cite{ethier86}. Unlike
the chemical master equation (\ref{masterg}), which gives a differential equation for the density of $X_i(t)$, the representation (\ref{sde}) can be analyzed using probabilistic techniques to obtain meaningful information, even for extremely complicated reaction networks. Additionally, since it is a pathwise representation rather than a description of a density it is more useful as a basis for deriving
simulation algorithms. While one usually thinks of (\ref{masterg}) or (\ref{sde}) as representing a well-mixed chemical reaction, this framework is easily extended to spatial models of reaction diffusion systems where the diffusion is represented by additional reaction channels \cite{Chevallier2016}. Hence, although we will usually speak of non-spatial reaction networks in this paper, our results can be extended to spatial models.\par

There are two main methods for generating exact realizations of the paths (\ref{sde}). These are the Next Reaction method, and the Gillespie algorithm. We will present only the Gillespie algorithm in, which is given in Algorithm \ref{alg_sde}, since it is more easily related to other algorithms used in this paper. In particular it can be modified to simulate piecewise  deterministic Markov process in a fairly straightforward manner.  
 Let us denote
\begin{align*}
\sfC_{X} &:= \expect[\text{ \# of computations to simulate $X(T)$ using Algorithm \ref{alg_sde}}].
\end{align*}
More generally, we define the \emph{complexity} of an algorithm as the order of magnitude of the expected simulation time with respect to some parameter in the model. For our proposes this will be the system size $\Omega$. To get a handle on $\sfC_X$, note that  $\expect[t_{\rm next}] =  O(\Omega^{-1})$ implying we expect to simulate $T\expect[t_{\rm next}]^{-1} = O(\Omega)$ events, so that $\sfC_{X} = O(\Omega)$. It should be noted that the most computationally expensive step in each iteration of Algorithm \ref{alg_sde} is the computation of the reaction index $k$, and the next reaction method can be modified to reduce the cost of this step \cite{sanft2015}. On the other hand, it is the $O(\Omega^{-1})$ scaling of the complexity that is important for our analysis. 

\begin{algorithm}
\caption{Simulation of (\ref{sde})}\label{alg_sde}
\begin{algorithmic}[1]
\State Initialize $X(0)$ and set $t = 0$.
\While{$t \le T$}
        \State $\alpha_{\rm total}:= \Omega\sum_{j}\alpha_j(X(s)/\Omega)$
        \State Generate random numbers $t_{\rm next} \sim \text{Exp}(\alpha_{\rm total})$ and $r \sim \text{Unif}(0,1)$
        \State $\bar{j} :=  \min_{i}\{i:\sum_{j=1}^i\alpha_j(X(s))<r\alpha_{\rm total}\}$ 
        \State $X(t+t_{\rm next} ) := X(t)+k_{\bar{j}}$
        \State $t := t+t_{\rm next} $
\EndWhile
\end{algorithmic}
\end{algorithm}

\subsection{Monte Carlo simulations in the classical setting}

Suppose we wish to approximate some statistics of some process $X_k(T)$ using a crude Monte Carlo estimator (MCE) $\widehat{Q}_{\rm crude}(M)$ with $M$ realizations of the process. That is,
\begin{equation}
\widehat{Q}_{\rm crude}(M) = \frac{1}{M}\sum_{j=1}^Mf(X_{k,[j]}(T)).\label{mce_crude}
\end{equation}
Here and elsewhere we use the convention that the subscript $[j]$ indicates a specific realization of a process. In order for $\widehat{Q}_{\rm crude}(M)$ to approximate $\expect[ f(X_k(T))]$ to order $\varepsilon$ in the sense of confidence intervals, that is, in the sense that the standard deviation of $\widehat{Q}_{\rm crude}(M)$ is order $\varepsilon$, we need
\begin{equation*}
\var(\widehat{Q}_{\rm crude}(M)) = O(\varepsilon^2)
\end{equation*}
and hence $M = O(\varepsilon^{-2}\var(f(X_{k,[j]}(T)))$. It is important to note that $\Omega$ and $\varepsilon$ must be related.  
 To see how such a relationship comes about, observe that in practice $\varepsilon$ will scale with ${\Omega}$ since we need to decrease $\varepsilon$ to resolve demographic noise as $\Omega$ increases. Following \cite{anderson2015b}, we will set
\begin{equation*}
\varepsilon = {\Omega}^{-\delta}
\end{equation*}
where $\delta \ge 0$ is a measure of the accuracy of the Monte Carlo estimator relative to the system's noise.
Then the complexity of the MCE $\widehat{Q}_{\rm crude}(M)$ is
\begin{equation*}
\sfC_{\rm crude}  = O(\varepsilon^{-2}\sfC_{X}\var(f(X_{k}(T))).
\end{equation*}
It can be shown that $\var(f(X_{k}(T)) = O(\Omega^{-1})$, which corresponds to the stochastic model approaching a deterministic limit at a rate inversely proportionally to the size of the system. 
On the other hand, $\sfC_{X} = O(\Omega)$, so that in the classical setting, the contribution of the complexity from the simulation of the path cancels with the variance of the path, and we obtain 
\begin{equation*}
\sfC_{\rm crude}  = O(\varepsilon^{-2}).
\end{equation*}
Essentially, when a stochastic model approaches a deterministic limit, variance reduction in terms of the system size is ``for free". This applies not only to the crude MCE, but any Monte Carlo method applied to a biochemical model in the classical setting. We refer to \cite{anderson2015b} for a rigorous analysis of Monte Carlo methods in the classical setting. The crucial observation that motivates the developments in this paper is that in the multiscale setting, where some species abundances do not scale with $\Omega$, it is generally not possible to bound $\var(f(X_{k}(T))$ in terms of $\Omega$, and the complexity of any Monte Carlo method increases by an order of magnitude.

\par

\section{Multiscale reaction networks and the partial thermodynamic limit }\label{sec_ptdl}

We want to consider multiscale chemical reaction networks involving $l$ + $m$ species, where there are $l$ abundant species $\{\mcX_1,\mcX_2,\dots,\mcX_l\}$ and $m$ rare species $\{\mcY_1,\mcY_2,\dots, \mcY_m\}$. 
We will take $X_i$ and $Y_i$ to represent the number of species $\mcX_i$ and $\mcY_i$ respectively. The motivation behind this notation is that when $m=0$ we retrieve the classical setting both notationally and conceptually.  In some applications one refers to $X = (X_i)_{i = 1,\dots, l} \in \naturals^l$ as the population and $Y = (Y_i)_{i = 1,\dots, m} \in \naturals^m$
as the environment \cite{Levien16}, and we will often adopt this terminology. In order to justify the labelings of rare and abundant species we need two main assumptions. First, we assume that only the initial concentration of the abundant species scales with the systems size.

\begin{ass}[separation of scales]\label{ass_seperation}
Let $\Omega$ denote the system-size introduced in Section \ref{sec_basics}. Then 
\begin{equation*}
  ||X(0)||_1 = O({\Omega}),\quad ||Y(0)||_1 = O(1)
\end{equation*}
Here and throughout the rest of this paper the $O(\cdot)$ terms are with respect to $\Omega$. Note that for any vector $x \in \reals^l$, $||x||_1 = O(g(\Omega))$ is equivalent to saying $||x||_{\infty} = \max_{j}|x_j| = O(g(\Omega))$ so that this assumptions implies $|X_i(0)| = O(\Omega)$ for each $i = 1,\dots,l$. 
\end{ass}

\noindent Physically, ${\Omega}$ can be thought to control the volume $V(\Omega)$ of a well-mixed domain such that the density of each abundant species and the number of each rare species (environment) are independent of the volume. We will assume that $V(\Omega)$ is linearly increasing in $\Omega$.

 Following the notation of section 2,
each reaction in  $\mcR = \{\mcR_j,\, j=1,\dots,p\}$ has the form
\begin{align}
\mcR_j: \sum_{i=1}^{m} k_{j,i}^{\text{in}}\mcX_{i}+
\sum_{i=1}^{l} u_{j,i}^{\text{in}}\mcY_{i}
 {\to}
 \sum_{i=1}^{m} k_{j,i}^{\text{out}}\mcX_{i}+
\sum_{i=1}^{l} u_{j,i}^{\text{out}}\mcY_{i}
\label{reaction}
\end{align}
The reaction vectors are given by 
\begin{align*}
k_j &= ( k_{j,1}^{\text{out}} - k_{j,1}^{\text{in}},\dots, k_{j,l}^{\text{out}} - k_{j,l}^{\text{in}})^T \in \integers^l\\
u_j &= (u_{j,1}^{\text{out}} - u_{j,1}^{\text{in}},\dots, u_{j,m}^{\text{out}} - u_{j,m}^{\text{in}})^T \in \integers^m.
\end{align*}
Suppose that only a proper subset of the reactions involve changes in the state of the rare species and define
\begin{equation*}
\mcJ_1 =  \{j  : u_{j,i'} \ne 0, \mbox{ for some } i' \in  \{1,\dots,m\}\},\quad \mcJ_0=\{1,\ldots,p\}-\mcJ_1.
\end{equation*}
That is, $\{R_j: j \in \mcJ_1\}$  is the set of reactions that produce or annihilate at least one of the rare species. As a further simplification, we assume that this set of reactions does not produce changes in the number of any abundant species. The random time change representation of networks of this form can then be written in the form
\begin{subequations}
\label{sde_unscaled}
\begin{align}
X_i(t) &= X_i(0) + \sum_{j\in \mcJ_0}\Pi_j \left(\int_0^t\Omega  \alpha_j(X(s)/\Omega,Y(s))ds \right)k_{j,i},\quad i = 1,\dots, l\\
Y_{i'}(t) &= Y_{i'}(0) + \sum_{j\in \mcJ_1}\Pi_j \left(\int_0^t \alpha_j(X(s)/\Omega,Y(s))ds \right)u_{j,i'}, \quad i' = 1,\dots, m.
\end{align}
\end{subequations}
where $\Pi_j$ are once again independent unit rate Poisson processes.  Note that only the number of each abundant species is scaled by the system size $\Omega$. We will always assume that the rate functions $\alpha_j(x,y)$ are polynomials in $x$ and Lipschitz continuous $x$ uniformly in $y$, that is, there exists $L_{\alpha_j}(K)$ such that for all $||x_1-x_2||<K$, 
\begin{equation}
|\alpha_j(x_1,y) - \alpha_j(x_2,y)| \le L_{\alpha_j}(K)|x_1-x_2|\label{lip}
\end{equation}
for all $y \in \naturals^m$. This is consistent with the law of mass action for the abundant species, that is,
\begin{equation}\label{mass_action}
\alpha_j(x,y) = \zeta_j(y) \prod_{i=1}^mx_i^{k^{\rm in}_{j,i}}.
\end{equation}
We will assume $\zeta_j (y)= O(1)$ with respect to all other parameters, although with a more sophisticated scaling than the one carried out below our analysis can be extended to the case where $\zeta_j$ vary over many orders of magnitude \cite{anderson2012}. Scaled population models of this sort appear in a number of applications, such as gene networks where the promoters have a fixed number of states but the abundance of a gene scales with other properties (such as the size) of the cell
\cite{hufton2015}.  We would hope that the scaled rates along with suitable restrictions on the dynamics ensure the initial separation of scales is preserved over bounded time intervals. This usually involves some notion of stability, or regularity, for the stochastic equations (\ref{sde_unscaled}) and rigorous results concerning the stochastic stability of continuous time Markov chains can established for general biochemical reaction models \cite{Chevallier2016}. The focus of this paper is not stochastic stability, but rather variance reduction in Monte Carlo simulations, so to keep our presentation self contained we will make the following assumption.

\begin{ass}[stability]\label{ass:stable}
 Let $(X(t),Y(t))$ satisfy equation (\ref{sde_unscaled}).
 Then $||Y(t)||_1 \le c$ for all $t<T$ and for all $\sigma\ge1$, 
 \begin{equation}
 \expect[||X(t)||_1^{\sigma}] <  ||X(0)||_1^{\sigma}B_{\sigma}(T) \label{momentbound}
 \end{equation}
 where $B_{\sigma}(t)$ is independent of $||X(0)||_1$. 
Moreover, if we define the stopping time
  \begin{equation}
  \tau_K = \inf_{t \ge 0}\{t:||X(t)||_1 \ge K||X(0)||_1\}, \label{tau_K}
  \end{equation}
  then for all $T >0$ there exists a positive constant $C(T)$ independent of $K$ and the initial conditions such that \begin{equation}
\prob(\tau_K < T) \le C(T)K^{-1}\label{exitbound}.
\end{equation}
\end{ass}

Note that the first bound simply says that the growth rates of the moments are not dependent on the system size. Now let us briefly explain the intuition behind the last statement in Assumption \ref{ass:stable}, which is less transparent.  We begin by noting that the most desirable condition would be the existence of $K$ such that $||X(t)||_1\le K||X(0)||_1$ for all $t\le 0$. This ensures $||X(0)||_1 = O(\Omega)$ and $||Y(0)||_1 = O(1)$, which is exactly what we would like to achieve. For a deterministic system this a sensible condition, but it is far too strong a restriction on the stochastic process $X(t)$. In fact, we can only guarantee sample paths of (\ref{sde_unscaled}) will be bounded with probability one when there is a conservation law in effect ($||X(t)||_1$  and $||Y(t)||_1$ are constant). This suggests we must settle for a more realistic assumption, one that quantifies the event of a sample path breaking these inequalities and ensures $K$ can be chosen to make these events rare. We achieve this in Assumption \ref{ass:stable} by bounding the probability that $\tau_K>T$. If the corresponding mean field process derived below has a trapping region we would naturally expect Assumption \ref{ass:stable} to hold since the deterministic dynamics arising in the limit $\Omega \to \infty$ are relatively bounded. More generally, assumption \ref{ass:stable} holds for any model with only linear growth, see Example \ref{ex:gene}.

 Under these assumptions, we can introduce the scaled population variables, or densities, $X^{\Omega}(t) = {\Omega}^{-1}X(t)$, ensuring that $||X^{\Omega}(t)||_1 = O(1)$ in some bounded interval. 
The scaling of the rates leads to the scaled dynamics
\begin{subequations}
\label{sde_scaled}
\begin{align}
X_i^{\Omega}(t) &= X_i^{\Omega}(0) + \sum_{j\in \mcJ_0} \Pi_j \left({\Omega}\int_0^t\alpha_j(X^{\Omega}(s),Y(s))ds \right){\Omega^{-1}}k_{j,i},\quad i = 1,\dots, l\\
Y_{i'}(t) &= Y_{i'}(0) + \sum_{j\in \mcJ_1} \Pi_j \left(  \int_0^t\alpha_j(X^{\Omega}(s),Y(s))ds \right)u_{j,i'}, \quad i' = 1,\dots, m.
\end{align}
\end{subequations}
In light of this scaling, it makes sense to take the thermodynamic limit ${\Omega} \to \infty$ only in the abundant species, following \cite{jahnke2012} we call this the \emph{partial thermodynamic limit} (PTDL). Setting $\lim_{\Omega \to \infty} X^{\Omega}(t) = Z(t)$ in some sense we would expect that $Z_i(t)$ satisfies (\ref{sde_scaled}a) with the stochastic integral for the abundant species by mean field dynamics. That is,
\begin{subequations}\label{pdmp}
\begin{align}
Z_i(t) &= Z_i(0) + \sum_{j\in \mcJ_0}\int_0^t\alpha_j(Z(s),Y(s))k_{j,i,},\quad i = 1,\dots, l\\
Y_{i'}(t) &= Y_{i'}(0) + \sum_{j\in \mcJ_1}\Pi_j \left(\int_0^t \alpha_j(Z(s),Y(s))ds \right)u_{j,i'}, \quad i' = 1,\dots, m.
\end{align}
\end{subequations}
Another way to express the dynamics of $Z_i(t)$  is through the differential equation
\begin{equation}\label{pdmp_de}
Z_i'(t) = \sum_{j\in {\mcJ_0}} \alpha_j(Z(s),Y(s))k_{j,i}, \quad i = 1,\dots, l
\end{equation}
which holds between discrete jumps. $Z(t)$ is an example of a piecewise deterministic Markov process (PDMP). A comprehensive review of such process along with more rigorous constructions can be found in \cite{Malrieu15}. While there are many technical statements one can make about various notions of the convergence $X^{\Omega}(t) \to Z(t)$, the technical results necessary for our analysis are contained in Sections \ref{sec_coupling} and  \ref{sec_couplingfb}. \par 
The PDMP (\ref{pdmp_de}) should be compared to the piecewise SDE one would obtain if instead of taking the fully  deterministic mean field approximation of $X^{\Omega}(t)$, one carried out a system-size expansion \cite{gardiner1985,bressloff2014book} to leading order. This piecewise It\^{o} SDE is given by
\begin{equation}
dZ_i(t) = \sum_{j\in \mcJ_0}\left(\alpha_j (Z(s),Y(s))dt
+\Omega^{-1/2}\sqrt{\alpha_j (Z(s),Y(s))}dW_j(t)\right)k_{ji}
\end{equation}
where $W_j(t)$ are independent Brownian motions. Previous studies have developed
hybrid frameworks for simulating multi-scale chemical reaction networks using this
system-size expansion \cite{duncan2016,Ganguly2014}. However, in contrast to the Monte Carlo methods discussed
in the following sections, these hybrid methods have an error which is fixed
in terms of $\Omega$.

\subsection{Generating sample paths of the PDMP}
Unlike the jump process (\ref{sde_scaled}), exact realizations of $Z(t)$ cannot be generated due to the continuous nature of (\ref{pdmp_de}). Instead one must pick an accuracy $h$ and compute an approximation $(Z^h(t),Y^h(t))$
using standard ODE methods between jumps, and then incorporate the stochastic jumps in some suitable way. Note that one only needs to include the stochastic contribution of reaction channels indexed by the set $\mcJ_1$. If the model has no feedback than one can compute $Y^h(T) = Y(T)$ using an exact algorithm, and then solve for $Z^h(T)$ using a suitable ODE method between jumps.  For models with feedback, the computation of $Z^h(T)$ is more subtle, since one needs to compute the jumps times while evolving the process $Z^h(T)$. A number of algorithms exist for accomplishing this task. Popular among those studying PDMP approximations of chemical reaction networks are splitting methods, which can be derived by applying operator splitting to the forward Kolmogorov equation for the PDMP \cite{Chevallier2016}. Like $\tau$-leaping, splitting methods obtain approximate jumps times by holding the jump rates constant in $Z(t)$ and evolving $Y(t)$ over some small time interval.  We have found a more direct approach, where the next jump time is computed exactly, to be more effective for the models considered in this paper. We leave a more systematic comparison of methods for solving the PDMP arising in thermodynamic limit to future studies.\par  
We now describe direct method. Given the state of the system, $(Z(t_0),Y(t_0)) = (z_0,y_0)$, at time $t_0$, it can be shown that the next jump time is given by 
\begin{equation}\label{truejump}
t_1 = \inf\left\{u>t_0 :\int_{t_0}^u \sum_{j \in \mcJ_2}\alpha_j(\Phi^s_{y_0}(z_0),y_0)ds =\bar{s}  \right\}
\end{equation}
where $\bar{s} \sim \text{Exp}(1)$ and $\Phi^t_{y}$ is the continuous flow of $Z(t)$ when $Y(t) = y$  \cite{Riedler2013}. In general, the process of solving for $(Z^h(t),Y^h(t))$ at the jump times by numerical solving (\ref{truejump}) is known as the True Jump Method. One can show that the solution to the minimization problem (\ref{truejump}) along with the state of the continuous variable at $t_1$ is given by $(Z(t_1^-),t_1) = (z(\bar{s}),\tau(\bar{s}))$ where 
\begin{align}\label{chv_ode}
\left\{\begin{array}{l}
z'(s) = \sum_{j \in \mcJ_1}\alpha_j(z(s),y_0)\left(\sum_{j \in \mcJ_2}\alpha_j(z(s),y_0)\right)^{-1}\\
\tau'(s) = \left(\sum_{j \in \mcJ_1}\alpha_j(z(s),y_0)\right)^{-1}\\
z(0) = z_0,\quad \tau(0) = t_0.
\end{array}\right.
\end{align}
The process of solving the PDMP using by (\ref{truejump}) to obtain the solutions to (\ref{chv_ode}) has been coined the CHV method \cite{veltz2015}. One subtly is that a direct implementation of the CHV method gives only the values of $Z(t)$ at the jump times, whereas we will be interested in obtaining the value of the process at some specified time $T$. However, this can easily be achieved by integrating the continuous component of the process from the last jump time before $T$, up to $T$. 
The details of the CHV method are given in Algorithm \ref{alg_pdmp}. At leading order in $h$, the complexity of obtaining $Z^h(T)$ with this algorithm is simply the complexity using whatever integration scheme we use to solve (\ref{chv_ode}) up to time $T$. In particular, using high order methods renders the cost of computing $Z^h(T)$ negligible relative to the cost of computing $X(T)$ exactly when $\Omega$ is large. 

%
%
\begin{algorithm}
\caption{Simulation of (\ref{pdmp})}\label{alg_pdmp}
\begin{algorithmic}[1]
\State Select an accuracy $h$. Initialize $Z^h(0)$ and set $t = 0$.
\While{$t \le T$}
        \State Generate a random number $\bar{s} \sim \text{Exp}(1)$.
        \State Let $(z(\bar{s}),\tau(\bar{s}))$ be the solution to (\ref{chv_ode}) with $t_0 = t$ and $z_0 = Z^h(t)$. 
	\State $Z^h(t+\tau(\bar{s})) = z(\bar{s})$
	\State $\alpha_{\rm total} = \sum_{j \in \mcJ_1}\alpha_j(Z^h(\tau(\bar{s})),Y(t))$
        	  \State Generate a random number $r \sim \text{Unif}(0,1)$
       	 	\State $\bar{j}:=  \min_{i}\{i:\sum_{j \in \mcJ_2:j<i}\alpha_j(Z^h(s),Y^h(s))<r\alpha_{\rm total}\}$
        		\State $Y^h(\tau(\bar{s})) := Y^h(t)+u_{\bar{j}}$
        	\State $t := \tau(\bar{s})$
\EndWhile
\State Perform numerical integration to obtain $Z^h(T)$. 
\end{algorithmic}
\end{algorithm}

\subsection{Example: Two State Catalyst }\label{ex_catalyst}%
Our first example of a simple chemical reaction network without feedback consists of four species  $\mcX_1,\mcX_2,\mcY_1$ and $\mcY_2$ obeying the reaction scheme
\begin{align}
\begin{split}
\mcX_1+\mcY_1&\underset{\alpha}{\to} \mcX_2+\mcY_1\\
\mcX_2 +\mcY_2&\underset{\beta}{\to} \mcX_1+\mcY_2\\
\mcY_1&\underset{\lambda}{\to} \mcY_2\\
\mcY_2&\underset{\mu}{\to} \mcY_1
\end{split}\label{example1}
\end{align}
where $X_1(0)+X_2(0)= {\Omega} \gg 1$ and $Y_1+Y_2 = 1$. This reaction network can model a catalytic reaction \cite{Chevallier2016}, or a population of particles jumping between two lattice sites with gate controlled by the $Y$ species \cite{Levien16}. In the former context $Y_i$ would be the catalysts, while in the latter they would represent the state of the barrier between two compartments and $X_i$ would represent the number of particles in each compartment. \par

In terms of Kurtz's time change representation,
\begin{align*}
X(t) &= X(0) -\Pi_1\left( \int_0^tX(s)Y(s) \alpha ds \right) + \Pi_2\left( \int_0^t(\Omega-X(s))(c-Y(s))   \beta ds \right) \\
Y(t) &= Y(0) - \Pi_3\left(\int_0^t \lambda (c-Y(s))ds \right)+ \Pi_4\left(\int_0^t \mu Y(s)ds \right)
\end{align*}
where we have used the conservation laws $X_2(t) = {\Omega} -X_1(t)$, $Y_2(t) = c -Y_1(t)$ and set $X_1(t) = X(t)$, $Y_1(t) = Y(t)$. The scaled model is
\begin{align*}
X^{\Omega}(t) &= X^{\Omega}(0) -\Pi_1\left(\Omega \int_0^tX^{\Omega}(s)Y(s) \alpha ds \right) \Omega^{-1}\\ &+ \Pi_2\left( \Omega\int_0^t(1-X^{\Omega}(s))(c-Y(s)) \beta ds \right) \Omega^{-1}\\
Y(t) &= Y(0) - \Pi_3\left(\int_0^t \lambda (c-Y(s))ds \right)+ \Pi_4\left(\int_0^t \mu Y(s)ds \right)
\end{align*}
In the limit $ {\Omega} \to \infty$ we obtain the PDMP
\begin{align*}
Z(t) &= Z(0) - \int_0^tZ(s)Y(s) \alpha ds +  \int_0^t(\Omega-Z(s))(c-Y(s))   \beta ds  \\
Y(t) &= Y(0) - \Pi_3\left(\int_0^t \lambda (c-Y(s))ds \right)+ \Pi_4\left(\int_0^t \mu Y(s)ds \right)
\end{align*}
Note that assumption \ref{ass:stable} holds due to the conservation laws.

\subsection{Example: Switching Gene}\label{ex:gene}%

A second example of a network without feedback is a simplified model of protein production in the presence of a switching gene \cite{hufton2015,Newby2012}:
\begin{align*}
\mcY_1&\underset{{\Omega}\alpha}{\to} \mcY_1+\mcX\\
\mcX&\underset{\beta}{\to}  \emptyset\\
\mcY_1&\underset{\lambda}{\to} \mcY_2\\
\mcY_2&\underset{\mu}{\to} \mcY_1
\end{align*}
 In this reaction network $\mcX$ represents a protein and $\mcY_i$ the state associated gene. When the gene is in state $\mcY_1$, the protein is produced at a rate proportional to the system size $\Omega$, while when in state $\mcY_2 $ the protein is not produced. For example, $\mcY_1$ could represent the presence of a repressor occupying the RNA polymerase binding site, thereby preventing the transcription of the gene \cite{Karmakar04,Thomas2014}. In this simplified model the complex mechanisms of transcription and translations are viewed as a ``back box" represented by the production rate $\alpha$. 
The stochastic equations are
\begin{align*}
X(t) &= X(0) +\Pi_1\left( \int_0^tY(s){\Omega} \alpha ds \right) - \Pi_2\left( \int_0^tX(s)   \beta ds \right) \\
Y(t) &= Y(0) - \Pi_3\left(\int_0^t\lambda(1- Y(s))ds \right)+ \Pi_4\left(\int_0^t \mu Y(s)ds \right)
\end{align*}
where we have once again used conservation to obtain a two dimensional system.
In the partial thermodynamic limit we obtain
\begin{align*}
Z(t) &= Z(0) + \int_0^tY(s) \alpha ds - \int_0^tZ(s)   \beta ds  \\
Y(t) &= Y(0) - \Pi_3\left(\int_0^t \lambda(1- Y(s))ds \right)+ \Pi_4\left(\int_0^t \mu Y(s)ds \right)\\
\end{align*}
In order to establish that Assumption \ref{ass:stable} one only needs to note that this assumption holds for a Poisson process which dominates $X(t)$. 

\subsection{Example: Switching Gene with Feedback}\label{ex:gene_fb}%
A more complex model of gene expression might include feedback between the protein associated with the gene and environment controlling the gene expression \cite{hufton2015}. For example, suppose protein $\mcX$ activates the repressor blocking the RNA polymerase. This leads to the reaction network
\begin{align*}
\mcY_1&\underset{{\Omega}\alpha}{\to} \mcY_1+\mcX\\
\mcX&\underset{\beta}{\to}  \emptyset\\
\mcX + \mcY_1&\underset{\Omega^{-1}\lambda }{\to} \mcX+ \mcY_2 \\
\mcY_2&\underset{\mu}{\to} \mcY_1
\end{align*}
The stochastic equations are
\begin{align*}
X(t) &= X(0) +\Pi_1\left( \int_0^tY_1(s) {\Omega}\alpha ds \right) - \Pi_2\left( \int_0^tX(s)   \beta ds \right) \\
Y(t) &= Y(0) - \Pi_3\left(\int_0^t \Omega^{-1}\lambda X(s) (1- Y(s))ds \right)+ \Pi_4\left(\int_0^t \mu Y(s)ds \right)
\end{align*}
The partial thermodynamic limit is 
\begin{align*}
Z(t) &= Z(0) + \int_0^tY(s) \alpha ds - \int_0^tZ(s)   \beta ds  \\
Y(t) &= Y(0) - \Pi_3\left(\int_0^t Z(t)\lambda(1- Y(s))ds \right)+ \Pi_4\left(\int_0^t \mu Y(s)ds \right).
\end{align*}
Since the $\mcJ_1$ channels have rates that depend on the abundant variables, we have a network with feedback. An argument similar to the one made for example \ref{ex:gene} implies that assumption \ref{ass:stable} holds.

\section{Variance reduction in the Multiscaling setting}\label{sec_mce}

As noted at the end of Section \ref{sec_basics}, for multiscale models we can generally not bound the sample variance in terms of the system size and the asymptotic complexity of Monte Carlo methods picks up a factor of $\Omega$. For example, 
\begin{equation*}
C_{\rm crude} = O(\varepsilon^{-2-1/\delta}). 
\end{equation*}
Loosely speaking, while a great deal of information about the exact model is still contained in the thermodynamic limit $Z(t)$, it is not being used in the computations. Again, we emphasis how this contrasts the classical setting where information about the deterministic limit is used to accelerate the converge of an MCE without any additional work. Generally speaking, our goal is to understand how information about the thermodynamic limit can be used in the multiscale setting. \par

The coupled Monte Carlo estimator we will introduce is based on the idea of variance reduction via a probabilistic coupling of the exact process with an approximate process. In our case the approximate process will be the PDMP $Z(t)$. This idea has proven to be very useful in, and is the basis for, multilevel Monte Carlo methods \cite{anderson2012,giles2015} where different $\tau$-leaping approximations are coupled. To construct an MCE in the present setting we note that
\begin{equation}
\expect [f(X_i(T))] = \expect[ f(X_i(T))- f(Z_i(T))] + \expect[f(Z_i(T))]\label{cmc_exact}
\end{equation}
Two observations allow us to use this decomposition to obtain statistics of $\expect [f(X_i(T))]$ more efficiently than the crude MCE (\ref{mce_crude}). First, statistics of $Z(T)$ can be obtained much more efficiently than statistics of the exact process when $\Omega$ is even moderately large relative to the abundance of the rare species.  These statistics can be obtained either by a MCE using Algorithm \ref{alg_pdmp} to generate the sample paths, or by non-Monte Carlo based methods given in Appendix \ref{sec_nonmce}, which are difficult to apply to the exact process.  Second, we can reduce the number of simulations we need to perform of the full process by coupling the processes $X(t)$ and $Z(t)$ in a way that reduces the variance of  the difference $f(X_i(T))- f(Z_i(T))$. \par

For the second term in (\ref{cmc_exact}) let us assume we can produce an approximation $\widehat{Q}_{Z}(h) \approx  \expect[f(Z_i(T))]$ satisfying 
\begin{equation*}
|\widehat{Q}_{Z}(h)- \expect[f(Z_i(t))]| = O(h)
\end{equation*}
and  $\var(\widehat{Q}_{Z}(h)) = O(h)$ with $h < \varepsilon$. 
We will also need an approximate path to estimate the term $f(X_i(T))- f(Z_i(T))$.  This will be denoted by $Z_i^h(T)$ and taken to satisfy
\begin{equation*}
\expect[|Z_i^h(T) - Z_i(T)|^2] =O(h^2).
\end{equation*}
Note that $h$ will generally depend on $\varepsilon$, and hence $\Omega$, so the use of the $O(\cdot)$ notation in this context is consistent with our earlier statement that asymptotic results are with respect to the system size. 
Such a bound will be satisfied provided the numerical integration step in algorithm \ref{alg_pdmp} is convergent. 
Then an $O(\varepsilon)$ estimator $\widehat{Q}_{\rm coupled}(M_1,h)$ of the form
\begin{equation}
\expect f(X_i(T)) = \expect[ f(X_i(T))- f(Z_i^h(T))] + \expect[f(Z_i^h(T))]\label{cmcApprox}
\end{equation}
can be constructed by summing the estimator
\begin{equation*}
\widehat{Q}_{(X,Z^h)}(M_1) = \frac{1}{M_1}\sum_{j=1}^{M_1}(f(X_{i,[j]}(T))-f(Z_{i,[j]}^h(T)))
\end{equation*}
and the approximation $\widehat{Q}_{Z}(h)$:
\begin{equation*}
\widehat{Q}_{\rm coupled}(M_1,h) = \widehat{Q}_{(X,Z^h)}(M_1)  + \widehat{Q}_{Z}(h).
\end{equation*}
Setting
\begin{equation}
V_{(X,Z_h)}(T) = \var( f(X_i(T))-f(Z_i^h(T))),\label{var}
\end{equation}
the variance of the coupled MC estimator, $V_{\text{coupled}}$, is simply
\begin{align}
V_{\text{coupled}} = M_1^{-1}V_{(X,Z_h)} + \var(\widehat{Q}_{Z}(h)) \sim M_1^{-1}V_{(X,Z^h)}^{\Omega} \label{var_coupled}
\end{align}
where we are assuming the first term is leading order in $\Omega$, meaning that the limiting factor in reducing the variance in the simulation of the coupled path. The methods of this paper are most applicable when $\widehat{Q}_{Z}(h)$ is computed with non-Monte Carlo based methods, and hence $\var(\widehat{Q}_{Z}(h)) = 0$. However, we do not exclude the possibility that $\widehat{Q}_{Z}(h)$ is computed from Monte Carlo simulations of the PDMP. In such cases we would assume that simulations of $Z^h(t)$ are much cheaper than simulations of $X(t)$ and hence $\var(\widehat{Q}_{Z}(h))$ can be made $o(\epsilon)$ at a negligible cost, for example, this will be the case when the search step in Algorithm \ref{alg_sde} is very expensive.
Of course for an order $\varepsilon$ estimator, we require $V_{\text{coupled}} = O(\varepsilon^{2})$ so that
\begin{equation*}
M_1 = \varepsilon^{-2}V_{(X,Z_h)}. 
\end{equation*}
It is now clear that for $\widehat{Q}_{\rm coupled}$ to be preferable over $\widehat{Q}_{\rm crude}$ we must be able to make $V_{(X,Z_h)}$ small and $\widehat{Q}_{Z}(h)$ must be cheap to generate. To be more precise,  
set
\begin{align*}
 \sfC_{(X,Z^h)} &\approx \expect[ \text{ \# of computations to simulate the coupled process $(X,Z^h)$  }]\\
\sfC_{Q_{Z,h}} & \approx \expect[\text{ \# of computations to compute $\widehat{Q}_{Z}(h)$ }]
\end{align*}
then the computational cost of the coupled MC estimator is
\begin{equation*}
\sfC_{\text{coupled}} = M_1\sfC_{(X,Z^h)} + \sfC_{\widehat{Q}_{Z}(h)}.
\end{equation*}
Generally we would expect $\sfC_{(X,Z^h)}$ to be comparable to $\sfC_{X}$. One the other hand, while $\sfC_{\widehat{Q}_{Z}(h)}$ depends heavily on the method we use to generate $\widehat{Q}_{Z}(h)$, it will generally be true that $\sfC_{\widehat{Q}_{Z}(h)} \ll \sfC_{X}$ because generating statistics of $Z^h(t)$ does not require implementing the large number of stochastic events needed for $X(t)$.  Moreover, as we have already mentioned it is often possible to bypass Monte Carlo methods entirely to obtain these statistics since the process $(Z(t),Y(t))$ my have a very low dimensional stochastic component.  When Monte Carlo simulations are bypassed we only need to compute $\widehat{Q}_{Z}(h)$, even if we are interested in doing computational experiments over a range of values of $\Omega$.

\section{Models without feedback}\label{sec_coupling}

For networks without feedback, we can simply couple the full stochastic model to the PDMP by generating a realization of
$Y(t)$ and using it to drive both the full model and the PDMP.  This simplifies the complexity analysis and also provides an upper bound on the achievable complexity of the coupled MCE for more general systems. Explicitly, the coupling is
\begin{subequations}\label{coupling_nofb}
\begin{align}
Y(t) &= Y(0) + \sum_{j \in \mcJ_1}\Pi_j \left(\int_0^t\alpha_j(Y(s))ds \right)u_j \\
X(t) &= X(0) + \sum_{j \in \mcJ_0}\Pi_j \left(\Omega \int_0^t\alpha_j(X(s)/\Omega,Y(s))ds \right)k_j\\
Z(t) &= Z(0) + \sum_{j \in \mcJ_0} \int_0^t\alpha_j(Z(s),Y(s))ds k_j
\end{align}
\end{subequations}
It is not difficult to see that (\ref{coupling_nofb}) is a probabilistic coupling in the sense that the marginal distributions of $(X,Z)$ are indeed the distributions for the uncoupled processes $X(t)$ and $Z(t)$.  In the present setting algorithm \ref{alg_pdmp} can be replaced by the simpler algorithm \ref{alg_nofb} which reuses the process $Y(t)$ in an more efficient manner. Algorithm \ref{alg_nofb} is also easier to implement with existing software since once can use preexisting packages to perform algorithm \ref{alg_sde} and the numerical integration. 

\begin{algorithm}
\caption{Simulation of (\ref{coupling_nofb})}\label{alg_nofb}
\begin{algorithmic}[1]
\State Generate $\{Y(t)\}_{0 \le t \le T}$ using algorithm \ref{alg_sde}.\label{alg_nofb_step1}
\State Generate $\{X(t)\}_{0 \le t \le T}$ using algorithm \ref{alg_sde} with rates computed from $\{Y(t)\}_{0\le t \le T}$.\label{alg_nofb_step2}
\State Generate $\{Z^h(t)\}_{0 \le t \le T}$ using deterministic numerical integration with rates computed from $\{Y(t)\}_{0\le t \le T}$.\label{alg_nofb_step3}
\end{algorithmic}
\end{algorithm}

The following Theorem, which is proved in Appendix \ref{ap:proof1}, gives a rigorous bound on the second moment of the pathwise error, which can be used to estimate the variance. 

\begin{thm}\label{main_thm}

Suppose Assumptions \ref{ass_seperation} and \ref{ass:stable} hold in a network without feedback and let $(Y(t),X^{\Omega}(t),Z(t))$ be given by (\ref{coupling_nofb}). Then
\begin{equation}\label{thm:pathbound}
\expect[||X^{\Omega}(T) - Z(T)||_1^2\bfone{\{\tau_K>T\}}]= O(\Omega^{-1}).
\end{equation}
\end{thm}

Theorem \ref{main_thm} implies an upper bound on the variance of the coupled paths.  To see this, note that Assumption  \ref{ass:stable} implies there exists $\bar{C}_2(T)$ indepndent of $K$ and $\Omega$  such that. 
\begin{equation*}
\expect[||X^{\Omega}(T) - Z(T)||_1^2\bfone{\{\tau_K<T\}}] \le \bar{C}_2(T)C(T)K^{-1}. 
\end{equation*}
Following an argument in \cite{Chevallier2016}, for each $\gamma \in (0,1)$ we can select $K$ such that 
\begin{align*}
\expect[||X^{\Omega}(T) - Z(T)||_1^2\bfone\{\tau_K< T \}] &\le  \gamma \expect[||X^{\Omega}(T) - Z(T)||_1^2]. 
\end{align*} 
It follows that 
\begin{equation*}
\expect[||X^{\Omega}(T) - Z(T)||_1^2]  \le \frac{1}{1-\gamma}O(\Omega^{-1}) = O(\Omega^{-1}). 
\end{equation*}

Letting $L_f$ denote the Lipschitz constant of $f$, we have 
\begin{align*}
V_{(X,Z^h)}^{\Omega} &= \var(f(X_i^{\Omega}(T)) - f(Z_i^h(T))) \\
&\le 2\expect[|f(X^{\Omega}_i(T)) - f(Z_i(T))|^2] + 2\expect[|f(Z^h_i(T)) - f(Z_i(t))|^2] \\
&\le 2L_f^2\expect[||X^{\Omega}(T) - Z(T)||^2] +2  L_f^2\expect[||Z^h(T)-Z(T)||^2].
\end{align*}
The first term is $O(h)$ (which is $o(\Omega^{-1})$) assuming our numerical integration scheme is convergent.

\subsection{Complexity analysis}
 When assumptions \ref{ass_seperation} and \ref{ass:stable} are satisfied in a network without feedback, 
 equation (\ref{var_coupled}) and Theorem \ref{main_thm} imply that
 \begin{equation*}
V_{\text{cmc}} \sim M_1^{-1}V_{(X,Z^h)}^{\Omega} = M_1^{-1}\Omega^{-1}.
\end{equation*}
To make this asymptotically order $\varepsilon^{2}$ we set 
\begin{equation*}
M_1 = \Omega^{-1}\varepsilon^{-2} = \varepsilon^{1/\delta - 2}
\end{equation*}
 The complexity of this estimator is then 
\begin{equation*}
\sfC_{\text{coupled}} = \varepsilon^{1/\delta - 2}\sfC_{(X,Z^h)} + \sfC_{\widehat{Q}_{Z}(h)}.
\end{equation*}
As noted earlier, the rare and abundant species can be simulated separately using algorithm \ref{alg_nofb}. Defining
\begin{align*}
\sfC_{Y} & := \text{Complexity of Step \ref{alg_nofb_step1} in algorithm \ref{alg_nofb}}\\
\sfC_{X|Y} & :=  \text{Complexity of Step \ref{alg_nofb_step2} in algorithm \ref{alg_nofb}}\\
\sfC_{Z^h|Y} & :=  \text{Complexity of Step \ref{alg_nofb_step3} in algorithm \ref{alg_nofb}}
\end{align*}
we have
\begin{equation*}
\sfC_{(X,Z^h)} =  \sfC_{X|Y}+ \sfC_{Z^h|Y} +\sfC_{Y}.
\end{equation*}
On the other hand, for the crude MC estimator each path is obtained from the first two steps in algorithm \ref{alg_nofb}, hence 
\begin{equation*}
\sfC_{\rm crude} =  \varepsilon^{-2}(\sfC_{X|Y} +\sfC_{Y}).
\end{equation*}
Therefore
\begin{align*}
\sfC_{\rm coupled}  &= \varepsilon^{1/\delta}\sfC_{\rm crude} + \varepsilon^{1/\delta-2} \sfC_{(Z^h|Y)} + \sfC_{\widehat{Q}_Z(h)}
\end{align*}

Using a high order numerical integration scheme (for example, a Runga-Kutta method) renders $\sfC_{Z^h|Y}$ negligible at leading order, while $\sfC_{X|Y} = O(\varepsilon^{-1/\delta})$.
 Since the rates affecting $Y(t)$ do not scale with $\Omega$,  we have $\sfC_{Y}  = O(1)$.
This means that 
\begin{equation*}
\sfC_{\rm crude}  = O(\varepsilon^{-2-1/\delta})
\end{equation*}
while
\begin{equation*}
\sfC_{\rm coupled} = O(\varepsilon^{-2}) +  \sfC_{Q_{Z,h}}.
\end{equation*}
We work under the realistic assumption that $\sfC_{Q_{Z,h}} = o(\varepsilon^{-2})$, and therefore $
\sfC_{\rm coupled} = O(\varepsilon^{-2})$.

\section{Extension to models with feedback}\label{sec_couplingfb}

Constructing an effective coupling for a model with feedback is more subtle. This is because we cannot simply reuse the realization $\{Y(t)\}_{0 \le t \le T}$, as this would generally not yield a probabilistic coupling. In other words, if we let $(X^{\Omega}(t),Y^X(t))$ and $(Z(t),Y^Z(t))$ be given by (\ref{sde_scaled}) and (\ref{pdmp}) where we have replaced $Y(t)$ with $Y^X(t)$ and $Y^Z(t)$ to emphasize these are different paths, it is generally not the case that $Y^X(t)$ and $Y^Z(t)$ are equivalent in any standard sense of equivalence for stochastic processes. Formulating the problem in this way suggests that the appropriate generalization of the coupling (\ref{coupling_nofb}) in this setting involves a coupling of $Y^X(t)$ and $Y^Z(t)$. Indeed, if we can construct a coupling that keeps $\Delta Y^{\Omega}(t) := Y^X(t)-Y^Z(t)$ small, then 
\begin{equation*}
(X^{\Omega}(t),Y^X(t)) = (X^{\Omega}(t),\Delta Y^{\Omega}(t)+Y^Z(t))
\end{equation*}
should remain close to $(Z(t),Y^Z(t))$, and when $\Delta Y^{\Omega}(t)  = 0$ for all $t\ge0$ we retrieve a coupling of the form (\ref{coupling_nofb}).
  \par

The method we will use to couple $Y^X(t)$ and $Y^Z(t)$ is known in the probability literature as the \emph{split coupling}  \cite{anderson2015} and has appeared in the context of coupling $\tau$-leaping approximations with different values of $\tau$ \cite{anderson2015ion,anderson2015}.
Let us introduce the split coupling in the context of a much simpler problem, namely the problem of coupling two homogenous Poisson processes with propensities $\alpha_1$ and $\alpha_2$. We continue to use the convention that $\Pi_i(t)$ are independent unit rate Poisson process, so the two processes we wish to couple can be written 
\begin{equation*}
Y_1(t) = \Pi_1(\alpha_1 t),\quad Y_2(t) = \Pi_2(\alpha_2 t). 
\end{equation*}
The idea of the split coupling is to break up these reaction channels into two parts, one that is common to both the processes, and another that is specific to a given process. In this way, the coupling of $Y_1(t)$ and $Y_2(t)$ will be driven by three Poisson processes: one counting the events that occur in both $Y_1(t)$ and $Y_2(t)$, one counting the events occurring only in $Y_1(t)$, and another counting the events occurring only in $Y_2(t)$. 
In order to write down the explicit representations of these counting processes we introduce the function 
\begin{equation*}
\rho(\alpha_1,\alpha_2) :=\alpha_1-  \alpha_2\wedge \alpha_1.
\end{equation*}
Now we observe that for two Poisson processes $\Pi_3(t)$ and $\Pi_4(t)$
\begin{equation*}
\Pi_3((\alpha_1 \wedge \alpha_2) t) + \Pi_4(\rho(\alpha_1,\alpha_2)t) 
\end{equation*}
is identical in distribution to $Y_1(t)$ due to the additive property of the Poisson processes.  
It should then be clear that (after introducing one additional Poisson processe $\Pi_5(\cdot)$), a coupled representation $(Y_1(t),Y_2(t))$ of $Y_1(t)$ and $Y_2(t)$ is given by 
\begin{align*}
Y_1(t)&=\Pi_3((\alpha_1 \wedge \alpha_2)t) + \Pi_4(\rho(\alpha_1,\alpha_2)t) \\
Y_2(t)&=\Pi_3((\alpha_1 \wedge \alpha_2)t) + \Pi_5(\rho(\alpha_2,\alpha_1)t)
\end{align*}
Here $\Pi_3(\cdot)$ is the common counting process described above that counts the events occurring in both $Y_1(t)$ and $Y_2(t)$, while 
$\Pi_4(\cdot)$ and $ \Pi_5(\cdot)$ count the events only occurring in one process. If $\alpha_1 \approx \alpha_2$, then $\rho(\alpha_1,\alpha_2)$ and 
$\rho(\alpha_2,\alpha_1)$ will be small, so that there are relatively few events uncommon to both process, and hence $Y_1(t)$ and $Y_2(t)$ remain close. \par

Returning to the problem of coupling $Y_X(t)$ and $Y_Z(t)$, the split coupling can be constructed in exactly the same manner for each reaction channel indexed by $\mcJ_1$. That is, for each $j \in \mcJ_1$ we replace two uncoupled counting processes of the form
\begin{align*}
&\Pi_{1,j}\left( \int_0^t \alpha_j(X^{\Omega}(s),Y^X(s))ds \right),\quad
\Pi_{2,j}\left( \int_0^t \alpha_j(Z(s),Y^Z(s))ds \right),
\end{align*}
where $\Pi_{1,j}(t)$ and $\Pi_{2,j}(t)$ are independent unit rate Poisson process, with 
\begin{align*}
&\Pi_{3,j}\left( \int_0^t \alpha_j(Z(s),Y^Z(s)) \wedge \alpha_j(X^{\Omega}(s),Y^X(s))ds  \right)\\ 
&\quad\quad \quad+\Pi_{4,j} \left(\Omega \int_0^t \rho(\alpha_j(X^{\Omega}(t),Y^X(s)),\alpha_j(Z(s),Y^Z(s)))ds \right)\\
&\Pi_{3,j}\left( \int_0^t \alpha_j(Z(s),Y^Z(s)) \wedge \alpha_j(X^{\Omega}(s),Y^X(s))ds  \right)\\ 
&\quad\quad \quad+\Pi_{5,j} \left(\Omega \int_0^t\rho( \alpha_j(Z(s),Y^Z(s)), \alpha_j(X^{\Omega}(t),Y^X(s)))ds \right)
\end{align*}
where $\Pi_{3,j}(t)$, $\Pi_{4,j}(t)$ and $\Pi_{5,j}(t)$ are all independent unit rate Poisson processes. 

After substituting  the coupled reaction channels into (\ref{sde_scaled}b) and (\ref{pdmp}b), with $\Pi_j\rightarrow \Pi_{1,j}$ and  $\Pi_j\rightarrow \Pi_{2,j}$, respectively,  we obtain their coupled representation 
\begin{subequations}\label{coupling_fb}
\begin{align}
\begin{split}\label{coupling_fb_a}
Y^X(t) &= Y(0) + \sum_{j\in \mcJ_1}\Pi_{3,j} \left(\int_0^t  \alpha_j( Z(s),Y^Z(s)) \wedge \alpha_j(X^{\Omega}(t),Y^X(s))ds \right)u_j\\
&+\sum_{j\in \mcJ_1}\Pi_{4,j} \left( \int_0^t\rho(  \alpha_j(X^{\Omega}(s),Y^X(s)),\alpha_j(Z(s),Y^Z(s)))ds \right)u_j\\
\end{split}\\
\begin{split}\label{coupling_fb_b}
Y^Z(t) &= Y(0) + \sum_{j\in \mcJ_1}\Pi_{3,j} \left(\int_0^t \alpha_j(Z(s),Y^Z(s)) \wedge \alpha_j(X^{\Omega}(t),Y^X(s))ds \right)u_j\\
&+ \sum_{j\in \mcJ_1}\Pi_{5,j} \left(\int_0^t \rho(\alpha_j(Z(s),Y^Z(s)),\alpha_j(X^{\Omega}(s),Y^X(s)))ds \right)u_j\\
\end{split}
\end{align}
\begin{align}
\label{coupling_fb_c}
X^{\Omega}(t) &= X^{\Omega}(0) + \sum_{j\in \mcJ_0}\Pi_j \left(\Omega  \int_0^t\alpha_j(X^{\Omega}(s),Y^X(s))ds \right)k_j \\
\label{coupling_fb_d}
Z(t) &= Z(0) + \sum_{j\in \mcJ_0} \int_0^t\alpha_j(Z(s),Y^Z(s))ds k_j.
\end{align}
\end{subequations}

 The structure of the coupled process is similar to that of the PDMP (\ref{pdmp}) in that there is a piecewise deterministic component, $Z(t)$, driven by a jump process whose evolution is governed by (\ref{coupling_fb_a}-\ref{coupling_fb_c}). In other words, (\ref{coupling_fb}) can be rewritten in the form of (\ref{pdmp}) by replacing $Y(t)$ with $(Y^X(t),Y^Z(t),X(t))$ and adjusting the rates accordingly. Viewing the coupled paths in this way is convenient since we already know how to simulate a generic PDMP. In particular, to generate sample paths of the coupled process (\ref{coupling_fb}) we simply use Algorithm \ref{alg_pdmp} with the appropriate replacements.  \par 

In order to extend Theorem \ref{main_thm} to the coupled processes in equation (\ref{coupling_fb}), recall that our original motivation for coupling $Y^X(t)$ and $Y^Z(t)$ was so that the resulting coupled process would look like a perturbation of the coupling (\ref{coupling_nofb}). Specifically, the perturbations come from the poisson processes $\Pi_{4,j}(\cdot)$ and $\Pi_{5,j}(\cdot)$, and if these terms vanish, $Y^X(t) = Y^Z(t)$, so that we retrieve the coupling of equation (\ref{coupling_nofb}). 
Using the process $\Delta Y^{\Omega}(t)$ defined above, we can write equations (\ref{coupling_fb_c}) and (\ref{coupling_fb_d})  as
\begin{align*}
X^{\Omega}(t) &= X^{\Omega}(0) + \sum_{j\in \mcJ_0}\Pi_j \left(\Omega  \int_0^t\alpha_j(X^{\Omega}(s),Y^X(s))ds \right)k_j \\
Z(t) &= Z(0) + \sum_{j \in \mcJ_0} \int_0^t\alpha_j(Z(s),Y^X(s)+\Delta Y^{\Omega}(s) )k_j ds .
\end{align*}
In the proof of Theorem \ref{main_thm} we relied on a Lipschitz bound  of the form (\ref{lip}). Recalling the properties of the rates $\alpha_j$, the appropriate generalization of the Lipschitz bound in the present context is
\begin{equation*}
|\alpha_j(x_1,y) - \alpha_j(x_2,y+\Delta y)| \le L_{\alpha_j}(K)||x_1-x_2||+ \sum_{p=1}^{\deg(\zeta_j(y))}a_{j,p} ||x_2||||\Delta y||^p.
\end{equation*}
where $a_{j,p}$ are constants independent of $\Omega$. 
Squaring this inequality,
\begin{align}\label{fb_ineq2}
\begin{split}
&|\alpha_j(x_1,y) - \alpha_j(x_2,y+\Delta y)|^2\\ &\quad \quad \quad \le L_{\alpha_j}(K)^2 ||x_1-x_2||^2+2\sum_{p=1}^{\infty}a_{j,p}L_{\alpha_j}(K)||x_1-x_2|| ||x_2||||\Delta y||^p + \left( \sum_{p=1}^{\deg(\zeta_j(y))}a_{j,p} ||x_2||||\Delta y||^p\right)^2
\end{split}
\end{align}
 It is clear that we now need to bound not only the moments of $||\Delta Y^{\Omega}(t)||_1$, as well as their products with  $||X^{\Omega}(t)-Z(t)||_1$, in order to obtain a bound on the complexity of the coupled Monte Carlo estimator. Note that $Z(t)$ is bounded independent of $\Omega$, so the $||x_2||$ terms in  (\ref{fb_ineq2}) do not play a significant role.  For simplicity, we focus on the case where $\zeta_j(y)$ is linear ($\deg(\zeta_j(y))=1$), while the general case will be analyzed in a forthcoming more technical paper. The following Lemma, which is proved in Appendix \ref{ap:proof2}, provides us with the result we need. 
\begin{lem} \label{lem:yfb}
Suppose Assumptions \ref{ass_seperation} and \ref{ass:stable} hold and let $(Y^X(t),Y^Z(t),X^{\Omega}(t),Z(t))$ be given by (\ref{coupling_fb}). Assume $\zeta_j(y)$ is linear for each $j \in \mcJ_2$. Then
\begin{subequations}\label{coupling_bound}
\begin{align}
\expect[||\Delta Y^{\Omega}(T)||_1^2\bfone{\{\tau_K>T\}}] &= O(\Omega^{-1/2})\label{coupling_bound1}\\
\expect[||X^{\Omega}(T)-Z(t)||_1||\Delta Y^{\Omega}(T)||_1\bfone{\{\tau_K>T\}}] &= O(\Omega^{-1/2})\label{coupling_bound2}
\end{align}
\end{subequations}
\end{lem}
Using Lemma \ref{lem:yfb} and (\ref{fb_ineq2}), we can obtain the following result.
\begin{thm}
Suppose Assumptions \ref{ass_seperation} and \ref{ass:stable} hold and let $(Y^X(t),Y^Z(t),X^{\Omega}(t),Z(t))$ be given by (\ref{coupling_fb}).  Moreover, assume $\zeta_j(y)$ is linear for each $j \in \mcJ_2$. Then

\begin{equation}\label{thm:pathbound_fb}
\expect[||X^{\Omega}(T) - Z(T)||_1^2\bfone{\{\tau_K>T\}}] =  O(\Omega^{-1/2}).
\end{equation}
\end{thm}
\begin{proof}
We get this result after inserting (\ref{fb_ineq2}) into the proof of Theorem \ref{thm:pathbound_fb} and applying Lemma \ref{lem:yfb}. \end{proof}

It follows that $V_{(X,Z^h)}^{\Omega}  = O(\Omega^{-1/2})$ when $X(t)$ and $Z^h(t)$ are coupled according to (\ref{coupling_fb}). 

\subsection{Complexity analysis}
In light of the analysis for the coupled MCE without feedback, and Theorem (\ref{thm:pathbound_fb}), we would expect that $O(\varepsilon^{-2-1/2\delta})$ complexity is achieved in the presence of feedback. To obtain this result one only needs to show that  $\sfC_{(X,Z^h)}  = O(\Omega)$
where the path is obtained from Algorithm \ref{alg_pdmp}. It is safe to assume that the contribution to the complexity of the numerical integration steps in Algorithm \ref{alg_pdmp} is $o(\Omega)$, but since we are solving for the coupled paths (\ref{coupling_fb}) the frequency of the jump events computed by solving \ref{chv_ode} scales with $\Omega$. These events then come to dominate the cost of this Algorithm so that we have $\sfC_{(X,Z^h)}  = O(\Omega)$, and as expected, $C_{\rm coupled} = O(\varepsilon^{-2-1/2\delta})$. Note that this result implies the computations are sped up by a factor of $O(\varepsilon^{-1/2\delta})$ in contrast to the  $O(\varepsilon^{-1/\delta})$ speed up when there is no feedback. 

\section{Numerical Examples}\label{sec:numerical}
To illustrate that our theoretical complexity results can be achieved in practice we have applied our coupled MCE to estimate $\expect[X^{\Omega}(20)]$ for a range of system sizes. The results for Examples \ref{ex_catalyst}, \ref{ex:gene} and \ref{ex:gene_fb} are displayed in Figure \ref{fig}, where we have used the couplings (\ref{coupling_nofb}) for the first two examples, and (\ref{coupling_fb}) for the third.\footnote{These computations were performed on a 2014 MacBook Pro with a 2.6 GHz Intel Core i5 processor using Python 3} For Examples \ref{ex_catalyst} we have solved for $\widehat{Q}_{Z}(h)$ by integrating the moment flow equations (\ref{moment_flow}), while in Examples \ref{ex_catalyst} and \ref{ex:gene_fb} we have used the solution to (\ref{dCK}) given in \cite{hufton2015}.

\begin{figure}[t!]
\centering
\includegraphics[width=15cm]{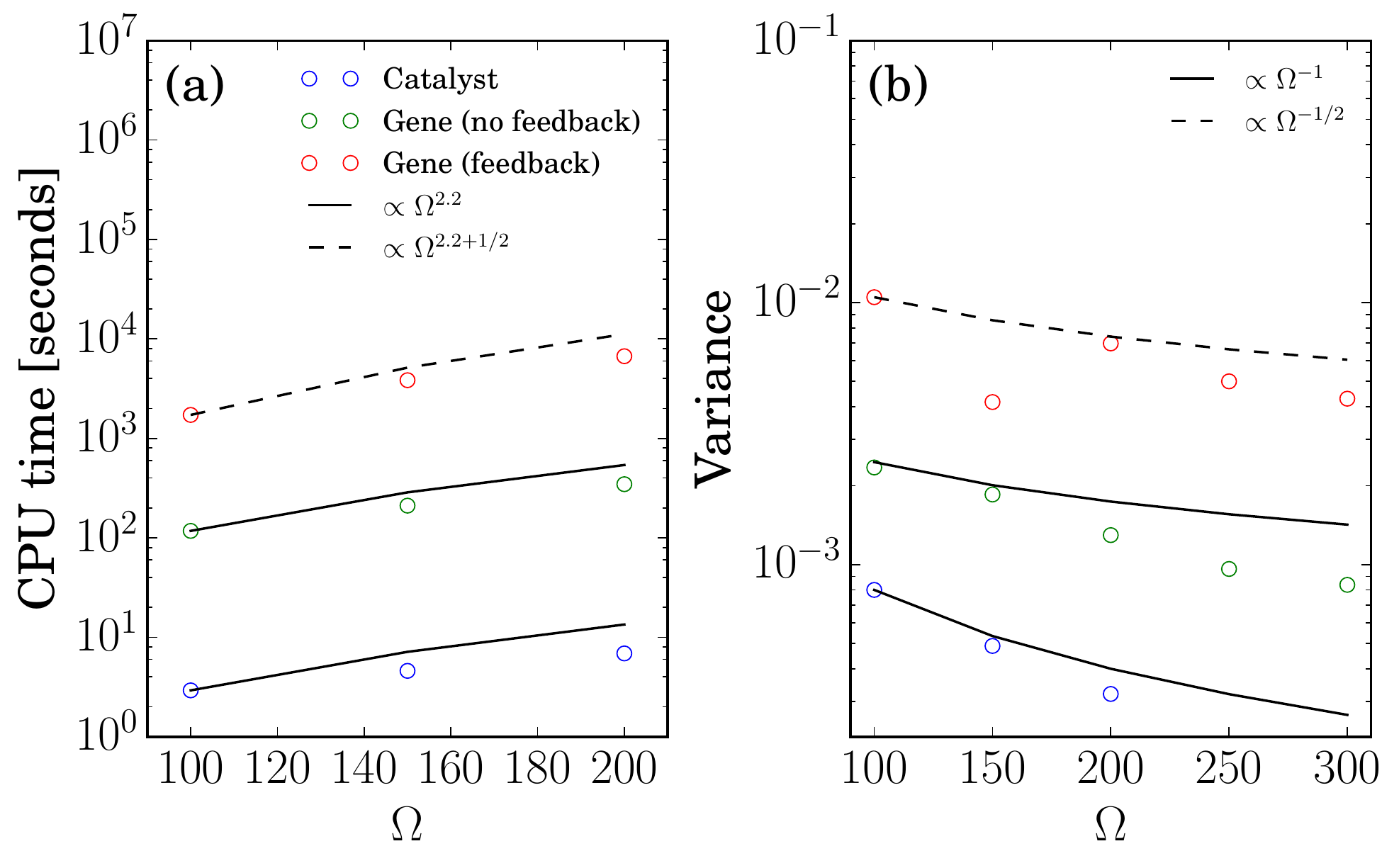}
\caption{ (a) The CPU time used to compute an estimate of $\expect[X^{\Omega}(10)]$ for each example, as well as the theoretical bounds. We have used $\delta = 1.1$. (b) The estimated variances using $10^5$ samples of the coupled paths, along with the theoretical bounds. 
Note the $\log$ scales on the time axis. We have used the LSODA algorithm to perform the numerical integration of (\ref{chv_ode}). Parameter values used are $\alpha = 1$, $\beta =1$, $\lambda =1.$ $\mu = 0.8$  for Example \ref{ex_catalyst}; and $\alpha = 1$, $\beta =1$, $\lambda =4.$ $\mu = 5.$ for  Examples \ref{ex:gene} and \ref{ex:gene_fb}. 
}\label{fig}
\end{figure}


\section{Conclusions}
Variance reduction in Monte Carlo estimators through probabilistic coupling has been used extensively in the scientific computing literature \cite{anderson2015ion,anderson2015,anderson2012,goodman2009}. However, there has been little work exploring the application of simplified models to reduce variances in Monte Carlo estimators for complex chemical reaction networks. We have extended the idea of variance reduction to models with partial thermodynamic limits in which the qualitative behavior of the full stochastic model is well approximated by a PDMP. Such population models arise in the biological and chemical sciences whenever the population can be decomposed into a group of abundant species, and a group of rare species.  The rare species often act as an environment that controls the dynamics of a large population, such as how the discrete state of a gene controls the production of a protein. When there is no feedback from the abundant species in the population, the coupling of the full model to the PDMP is straightforward. and we have shown that the coupling is indeed effective in reducing the variance of the Monte Carlo estimator. We have also derived an effective coupling for models with feedback, which applies to an extremely wide class of population models.  \par
 
Our results suggest that approximate stochastic models, such as the ones studied rigorously in \cite{jahnke2012,crudu2009,Chevallier2016} may be useful in the context of variance reduction for exact models. It would be particularly fruitful to extend our work to develop computational tools that are specifically tailored to spatial process. In particular, the reaction diffusion master equation (RDME) is a continuous time Markov chain approximation of   reaction diffusion processes for which there is a great deal of interest in simulating efficiently \cite{Isaacson2009a,Isaacson2009}. Other future directions include extending the coupling to other model reductions, such as the quasi-steady state, an idea that was briefly explored in \cite{anderson2012}.    \par

\section*{Acknowledgements}

This work was supported by the National Science Foundation (DMS 1613048 and RTG 1148230)

\begin{appendix}

\appendix

\section{Non-Monte Carlo based methods for a PDMP}\label{sec_nonmce}

The efficiency of the Monte Carlo methods introduced in this paper depend on our ability to obtain statistics of the PDMP without resorting to Monte-Carlo based methods. In this section we give a brief overview of the non-Monte Carlo based analysis of the PDMP.  All non-Monte Carlo based methods are derived from the forward Kolmogorov equation. While the Kolmogorov equation for the full model may be prohibitively complex due to the combinatorial complexity when there are many particles, the PDMP is often very simple.  We note that the problem of developing efficient numerical methods for solving the Kolmogorov equation for the full model and the PDMP is a very active area of research, but beyond the scope of this work. \par
The main object of interest will be the generator of the process $(Z(t),Y(t))$,
\begin{align}
\begin{split}
A f(z,y) &=\left( \sum_{j \in \mcJ_0} \alpha_j(z, y) k_j \right)\cdot \nabla_{z} f(z,y) \\
&+\sum_{j \in \mcJ_1} \alpha_j(y-u_j)f(z,y-u_j) -\alpha_j(y)f(z,y).
\end{split}
\end{align}

\subsection{Differential Chapman Kolmogorov equations}
The forward Kolmogorov equation, or \emph{differential Chapman Kolmogorov} (dCK) equation
\begin{equation}\label{dCK}
\frac{d}{dt}p(z,y,t)= A^{\dagger} p(z,y,t)
\end{equation}
gives the evolution of the density 
\begin{equation*}
p(z,y,t) = \prob(Z(t) = z,Y(t) = y). 
\end{equation*}
Here  the superscript ${\dagger}$ denotes the Hilbert adjoint operator with respect to the appropriate inner product. The dCK equations are a system of hyperbolic PDEs, and will generally only be useful for numerical simulation when they are finite dimensional. This occurs exactly when the rare species are conserved ($||Y||$ is constant) and the methods presented in this paper apply when this condition holds. The situation is also tricky when the abundant species are not conserved because the domain of $(\ref{dCK})$ may be unbounded. However, we can often solve the PDE on a bounded domain provided the deterministic dynamics have a trapping region. For example, see \cite{hufton2015}.

For example \ref{ex_catalyst} with $c=1$ the generator is
\begin{align*}
Af(z,0) &= -\alpha z\partial_z f(z,0)   - \lambda f(z,0)  + \mu f(z,1)\\
Af(z,1) &= \beta (1-z)\partial_z f(z,1)  +\lambda f(z,0)  - \mu f(z,1)
\end{align*}
so the dCK equation is
\begin{align*}
\frac{d}{dt}p(z,0,t) &=\alpha \partial_z [zp(z,0,t)]- \lambda p(z,0,t)  + \mu p(z,1,t)\\
\frac{d}{dt}p(z,1,t) &=- \beta \partial_z [(1-z) p(z,1,t)]+\lambda p(z,0,t)  - \mu p(z,1,t),
\end{align*}
along with the reflecting boundary conditions 
\begin{equation*}
p(0,y,t) = p(1,y,t)=0.
\end{equation*}

%

\subsection{Moment flow equations}
While solving (\ref{dCK}) numerically is preferable compared to Monte Carlo simulations of the PDMP (\ref{pdmp}) for chemical reaction networks, significant  computational gains can be made when one is only interested in expected values of polynomial function of the process $Z(t)$. This is done by deriving moment equations which are ODEs for the dynamic evolution of the moments and point correlation of $(Z(t),Y(t))$. We derive these from the backward Kolmogorov (bK) equation
\begin{align}
\frac{d}{dt}\expect[f(Z(t),Y(t))] = \expect[A f(Z(t),Y(t))].\label{bK}
\end{align}
For notationally simplicity we derive the moment equations when $z \in \reals$, although this derivation can easily be generalized \cite{hespanha2005}.  It is not difficult to see that the polynomial (\ref{mass_action}) is invariant under $A$.  In particular, if we set $f_{y'}^{(r)}(z,y) = z^r\bfone{\{y=y'\}}$ and write the polynoomial in the form
\begin{equation*}
\alpha_j(x,y) = \zeta_j(y)\sum_{i = 1}^{\deg(\alpha_j)}\alpha_{j,i}x^i,
\end{equation*}
then
\begin{align*}
A f_{y'}^{(r)}(z,y) &= \sum_{j \in \mcJ_0} \alpha^r_j(z, y) k_j  rf_{y'}^{(r-1)}(z,y) 
+ \sum_{j \in \mcJ_1} \alpha_j(y-u_j)f_{y'}^{(r)}(z,y-u_j)  -\alpha_j(y)f_{y'}^{(r)}(z,y)\\
&= \sum_{j \in \mcJ_0} \zeta_j(y) rk_{j}\sum_{i= 1}^{\deg(\alpha_j)} \alpha_{j,i}f_{y'}^{(r-1+i)}(z,y) \\
&+ \sum_{j \in \mcJ_1}\zeta_j(y-u_j)\sum_{i= 1}^{\deg(\alpha_j)} \alpha_{j,i}f_{y'}^{(r+i)}(z,y-u_j)  -\zeta_j(y)\sum_{i= 1}^{\deg(\alpha_j)} \alpha_{j,i}f_{y'}^{(r+i)}(z,y)
\end{align*}
This expression for $A f_{y'}^{(r)}(z,y) $ is particularly useful when $A$ preserves the degree of polynomials. This will happen when there is no feedback and $\deg(\alpha_j)= 1$. 
Then, setting $m^{(r)}_{y}(t) = \expect[f_{y}^{(r)}(Z(t),Y(t)) ]$, taking the exception of $A f_{y'}^{(r)}(Z(t),y)$ with respect to $Z$, and applying the bK equation  (\ref{bK}) yields the closed system of linear differential equations,
\begin{equation}\label{moment_flow}
\begin{split}
\frac{d}{dt}m^{(r)}_{y}(t)  &= \sum_{j \in \mcJ_0} \zeta_j(y) rk_{j}\sum_{i\in\{0,1\}}\alpha_{j,i}m_{y}^{(r-1+i)}(t)\\
&+\sum_{j \in \mcJ_1}\zeta_j(y-u_j)m_{y+u_i}^{(r)}(t) -\zeta_j(y)m_{y}^{(r)}(t)
\end{split}
\end{equation}
which we refer to as the \emph{moment flow equations}. 
If $\deg(\alpha_j)>1$ or the reaction network has feedback, then the system is not closed and one has to resort to moment closure techniques, which generally approximate the infinite-dimensional system of linear equations by a finite-dimensional system of nonlinear equations. See for example \cite{hespanha2005,singh2007derivative} for one approach. The efficiency of implementing a moment closure scheme depends heavily on the specific problem, and a detailed discussion of the computational issues associated with solving coupling moment equations is beyond the scope of this work.
\par
One complication that arises in deriving moment equations is that the conservation rules used to derive the dCK equation may be applicable to the moment equations. For example, taking $c=1$ in example \ref{ex_catalyst} we need to derive the moment flow equations for the first moments 
\begin{equation*}
m^{(1)}_{i,j}(t) = \expect[X_i(t)\bfone{Y(t) = j}]
\end{equation*}
where $X_i(t) = i +X(t)(1-2i)$, $i\in \{0,1\}$. 
These are
\begin{align*}
\frac{d}{dt}m^{(1)}_{0,0}(t)  &= -\alpha m_{0,0}^{(1)}(t) + \lambda m_{0,1}^{(1)}(t) - \mu m_{0,0}^{(1)}(t)\\
\frac{d}{dt}m^{(1)}_{0,1}(t)  &=  \beta m_{1,1}^{(1)}(t)   - \lambda m_{0,1}^{(1)}(t) + \mu m_{0,0}^{(1)}(t)\\
\frac{d}{dt}m^{(1)}_{1,0}(t)  &= \alpha m_{0,0}^{(1)}(t) + \lambda m_{1,1}^{(1)}(t) - \mu m_{1,0}^{(1)}(t)\\
\frac{d}{dt}m^{(1)}_{1,1}(t)  &=  -\beta m_{1,1}^{(1)}(t)   - \lambda m_{1,1}^{(1)}(t) + \mu m_{1,0}^{(1)}(t).
\end{align*}
The original system has two conservation rules but the moment equations only have one, so the system above could be reduced to a three dimensional system, but we cannot obtain a two dimensional system of moment equations from the two dimensional dCK equation. Instead we have derived these from the dCK equations without conservation.

\section{Proof of Theorem \ref{thm:pathbound} }\label{ap:proof1}
Throughout this Appendix $||\cdot|| = ||\cdot||_1$ and $t < T$ is fixed. Let $L_{\alpha_j}(K)$ denote the maximum Lipschitz constant of $\alpha_j(\cdot,y)$ over all values of $y$, which is independent of $\Omega$ but depends quadratically on $K$.
Define the constants
\begin{align*}
B_1&= \sum_{j\in \mcJ_0}||k_j|| L_{\alpha_j}(K) \\
B_2&= \max_j ||k_j|||\mcJ_0|
\end{align*} and let $\widehat{\Pi}_i: = \Pi_i(t) - t$ denote the centered Poisson process. Note that both $B_1$ and $B_2$ are independent of $\Omega$.
\begin{proof}[Proof of Theorem \ref{thm:pathbound}]
 Using the triangle inequality and adding and subtracting the appropriate terms, we have
\begin{align*}
||X^{\Omega}(t) - Z(t)|| &\le \max_{j\in \mcJ_0} ||k_j|| 
\sum_{j \in \mcJ_0} \left|\Omega^{-1} \Pi_j\left(\Omega \int_0^t \alpha_j(X^{\Omega}(s),Y(s))ds \right) - \int_0^t \alpha_j(Z(s),Y(s))ds \right|\\
&=  \max_{j\in \mcJ_0} ||k_j||\sum_{j \in \mcJ_0}  \Bigl|  \int_0^t \alpha_j(X^{\Omega}(s),Y(s)) - \alpha_j(Z(s),Y(s))ds \\
 &\quad\quad\quad +\Omega^{-1}\widehat{ \Pi}_j\left(\Omega \int_0^t \alpha_j(X^{\Omega}(s),Y(s))ds \right)\Bigl| \\
 &\le  \max_{j\in \mcJ_0} ||k_j||  \sum_{j \in \mcJ_0} \int_0^t| \alpha_j(X^{\Omega}(s),Y(s)) - \alpha_j(Z(s),Y(s))|ds\\
 &\quad\quad\quad  +\max_{j\in \mcJ_0} ||k_j||  \sum_{j \in \mcJ_0}\Omega^{-1}\left|\widehat{ \Pi}_j\left(\Omega \int_0^t \alpha_j(X^{\Omega}(s),Y(s))ds \right)\right|\\
 &\le  \max_{j\in \mcJ_0} ||k_j||  \sum_{j \in \mcJ_0} \int_0^tL_{\alpha_j}(K)|| X^{\Omega}(s)- Z(s)||ds\\
 &\quad\quad\quad + \max_{j\in \mcJ_0} ||k_j||  \sum_{j \in \mcJ_0}\Omega^{-1}\left|\widehat{ \Pi}_j\left(\Omega \int_0^t \alpha_j(X^{\Omega}(s),Y(s))ds \right)\right|\\ 
  &\le  B_1 \int_0^t|| X^{\Omega}(s)- Z(s)||ds+
  \max_{j\in \mcJ_0}||k_j||  \sum_{j \in \mcJ_0}\Omega^{-1}\left|\widehat{ \Pi}_j\left(\Omega \int_0^t \alpha_j(X^{\Omega}(s),Y(s))ds \right)\right|.
\end{align*}
Then squaring both sides of our inequality and applying Cauchy-Schwarz yields 
\begin{align*}
||X^{\Omega}(t) - Z(t)||^2& \le 3B_1^2t\int_0^t|| X^{\Omega}(s)- Z(s)||^2ds \\
&\quad\quad\quad + 3B_2^2\Omega^{-2}\sum_{j \in \mcJ_0}\left(\widehat{ \Pi}_j\left(\Omega \int_0^t \alpha_j(X^{\Omega}(s),Y(s))ds \right)\right)^2.
\end{align*}
Since the quadratic variation of the Martingale $\widehat{\Pi}_i(\cdot)$ is $\Pi_j(\cdot)$ and $\widehat{\Pi}_i(\cdot)^2-\Pi_j(\cdot)$ is Martingale \cite{ethier86}, the optional stopping theorem implies 
\begin{equation*}
\expect\left[\left(\widehat{ \Pi}_j\left(\Omega \int_0^t \alpha_j(X^{\Omega}(s),Y(s))ds \right)\right)^2\right]
=\expect\left[ \Pi_j\left(\Omega \int_0^t \alpha_j(X^{\Omega}(s),Y(s))ds \right)\right],
\end{equation*}
and noting that $ t \le  \tau_K$ an application of the smoothing formula \cite{bremaud2013} gives 
\begin{align*}
\expect\left[ \Pi_j\left(\Omega \int_0^t \alpha_j(X^{\Omega}(s),Y(s))ds \right)\right]
&\le  \expect\left[ \Pi_j\left(\Omega t L_{\alpha_j}(K) K\right)\right] =\Omega t L_{\alpha_j}(K) K 
\end{align*}
where $K$ is given in (\ref{tau_K}). Now set 
\begin{equation*}
B_3 = 4T^2\sum_{j \in \mcJ_0} L_{\alpha_j}(K).
\end{equation*}
Taking the expectation of the original expression and applying Fubini's theorem we get 
\begin{equation*}
\expect[||X^{\Omega}(t) - Z(t)||^2] \le   B_1 \int_0^t\expect[|| X^{\Omega}(s)- Z(s)||^2]ds
+  \Omega^{-1}B_2^2B_3K. 
\end{equation*}
After noting that all constants are $O(1)$, the result follows from Gr{\"o}nwall's inequality. 
\end{proof}

\section{Proof of Lemma \ref{coupling_bound}}\label{ap:proof2}

In this Appendix $a \le_C b$ will indicate that $a \le C b$ for some constant $C$ that is independent of $\Omega$ (but may depend on $K$ and $t$). Note that we are referring to different constants $C$ in each use of this notation.  In order to increase readability we will not keep track of the specific values of $C$, but the proof of Theorem \ref{thm:pathbound} gives a blueprint for how these constants can be derived. 
In the following proofs we will refer to the random variables
\begin{align*}
A_{1,j}(t) &:=  \int_0^t\rho(  \alpha_j(X^{\Omega}(s),Y^X(s)),\alpha_j(Z(s),Y^Z(s)))ds\\
A_{2,j}(t) &:=  \int_0^t \rho(\alpha_j(Z(s),Y^Z(s)),\alpha_j(X^{\Omega}(s),Y^X(s)))ds. 
\end{align*}
Of course 
\begin{equation*}
\Delta Y^{\Omega}(t) =  \sum_{j \in \mcJ_2}\left(\Pi_{4,j}(A_{1,j}(t) ) -  \Pi_{5,j}(A_{2,j}(t) )\right)u_j.
\end{equation*}
%

%

\begin{lem}\label{lem:a1}

Suppose Assumptions \ref{ass_seperation} and \ref{ass:stable} hold and let $(Y^X(t),Y^Z(t),X^{\Omega}(t),Z(t))$ be given by (\ref{coupling_fb}). Moreover, assume $\zeta_j(y)$ is linear for each $j \in \mcJ_2$. Then 
\begin{equation*}
\expect[||X^{\Omega}(t) - Z(t)||\bfone\{\tau_K> T\}] \le_C\int_0^t\expect[||\Delta Y^{\Omega}(s)||]ds +O(\Omega^{-1/2})
\end{equation*}
\end{lem}
\begin{proof}
Take $t <\tau_K$. Following the proof of Theorem \ref{thm:pathbound}, we have
\begin{align*}
||X^{\Omega}(t) - Z(t)|| 
  &\le_C  \int_0^t|| X^{\Omega}(s)- Z(s)||ds+\int_0^t||\Delta Y^{\Omega}(s)||||Z(s)||ds\\
  &\quad\quad\quad +  \sum_{j \in \mcJ_0}\Omega^{-1}\left|\widehat{ \Pi}_j\left(\Omega \int_0^t \alpha_j(X^{\Omega}(s),Y(s))ds \right)\right|.
\end{align*}

Using Jenson's inequality and the Martingale stopping theorem, 
\begin{align*}
&\expect \left[\left|\widehat{ \Pi}_j\left(\Omega \int_0^t \alpha_j(X^{\Omega}(s),Y(s))ds \right)\right| \right]\\
&\quad\quad\quad \le \expect \left[\left|\widehat{ \Pi}_j\left(\Omega tL_{\alpha_j}Kt\right)\right| \right]\le
\left( \expect \left[\left(\widehat{ \Pi}_j\left(\Omega tL_{\alpha_j}Kt\right)\right)^2\right]\right)^{1/2} \\
 &\quad\quad\quad= \left( \expect \left[\Pi_j\left(\Omega tL_{\alpha_j}Kt\right)\right]\right)^{1/2} = O(\Omega^{1/2})
 \end{align*}
from which is follows that
\begin{align*}
\expect[||X^{\Omega}(t) - Z(t)||]
  &\le_C \int_0^t\expect[|| X^{\Omega}(s)- Z(s)||]ds+\int_0^t\expect[||\Delta Y^{\Omega}(s)||||Z(s)||]ds +O(\Omega^{-1/2})
  \end{align*}
  so that the result is obtained from Gr{\"o}nwall's inequality and the fact that $Z(t)$ is bounded and independent of $\Omega$.
  \end{proof}

\begin{lem}\label{lem:a2}
Suppose Assumptions \ref{ass_seperation} and \ref{ass:stable} hold and let $(Y^X(t),Y^Z(t),X^{\Omega}(t),Z(t))$ be given by (\ref{coupling_fb}). Moreover, assume $\zeta_j(y)$ is linear for each $j \in \mcJ_2$. Then
\begin{equation*}
\expect[||\Delta Y^{\Omega}(t)||\bfone\{\tau_K> T\}] \le_C  \int_0^t\expect [||X^{\Omega}(s)-Z(s)||]ds.
\end{equation*}
\end{lem}
\begin{proof}
Take $t <\tau_K$.
\begin{align*}
||\Delta Y^{\Omega}(t)||_1& \le_C\sum_{j \in \mcJ_1}|\Pi_{4,j}\left(A_{1,j}(t) \right) - \Pi_{5,j}\left(A_{2,j}(t)\right)|
\end{align*}
Since the Poisson processes are positive, 
\begin{align*}
\expect[|\Pi_{4,j}\left(A_{1,j}(t) \right) - \Pi_{5,j}\left(A_{2,j}(t)\right)|]  &\le \expect[\Pi_{4,j}\left(A_{1,j}(t) \right)+ \Pi_{5,j}\left(A_{2,j}(t)\right)] \\
&= \expect[A_{1,j}(t)+A_{2,j}(t)].
\end{align*}
Now observe that for any $\alpha_1,\alpha_2>0$,  
\begin{align*}
\rho(\alpha_1,\alpha_2) + \rho(\alpha_2,\alpha_1) &= \alpha_1 - \alpha_1 \wedge \alpha_2+ \alpha_2 - \alpha_1 \wedge \alpha_2\\
&= \alpha_1+\alpha_2-2 \alpha_1 \wedge \alpha_2\\
&\le 2( \alpha_1 \vee \alpha_2 -  \alpha_1 \wedge \alpha_2) = 2|\alpha_1 - \alpha_2|,
\end{align*}
and hence 
\begin{align*}
\begin{split}
&A_{1,j}(t) + A_{2,j}(t)\le   \int_0^t \rho(  \alpha_j(X^{\Omega}(s),Y^X(s)),\alpha_j(Z(s),Y^Z(s)))\\
&\quad\quad\quad\quad\quad\quad\quad\quad+ \rho(\alpha_j(Z(s),Y^Z(s)),\alpha_j(X^{\Omega}(s),Y^X(s)))ds
\end{split}\\
&  \quad\quad\quad\quad\quad\quad \le 2\int_0^t| \alpha_j(X^{\Omega}(s),Y^X(s)) - \alpha_j(Z(s),Y^Z(s)))|\\
&\quad\quad\quad \quad\quad\quad\le_C  \int_0^t||X^{\Omega}(s)-Z(s)||ds+\int_0^t||\Delta Y^{\Omega}(s)||||Z(s)||ds.
\end{align*}
Consolidating our findings, 
\begin{align*}
\expect[||\Delta Y^{\Omega}(t)||]& \le_C  \int_0^t\expect [||X^{\Omega}(s)-Z(s)||]+\int_0^t\expect[||\Delta Y^{\Omega}(s)||||Z(s)||]ds.\end{align*}
Gr{\"o}nwall's inequality once again yields the final result.
\end{proof}

\begin{thm}\label{thm:a1}
Suppose Assumptions \ref{ass_seperation} and \ref{ass:stable} hold in a network with feedback and let $(Y^X(t),Y^Z(t),X^{\Omega}(t),Z(t))$ be given by (\ref{coupling_fb}).  Moreover, assume $\zeta_j(y)$ is linear for each $j \in \mcJ_2$. 
\begin{subequations}
\begin{align}
\expect[||\Delta Y^{\Omega}||_1\bfone{\{\tau_K>T\}}] &= O(\Omega^{-1/2})\label{coupling_bound_a1}\\
\expect[||X^{\Omega}(T) - Z(T)||_1\bfone{\{\tau_K>T\}}]& = O(\Omega^{-1/2}) \label{coupling_bound_a2}
\end{align}
\end{subequations}
\end{thm}
\begin{proof}
Combing Lemma \ref{lem:a1} and  Lemma \ref{lem:a2}, we have
\begin{equation*}
\expect[||\Delta Y^{\Omega}(t)||_1] \le_C \int_0^t\expect[||\Delta Y^{\Omega}(s)||_1]ds +O(\Omega^{-1/2})
\end{equation*}
so that Gr{\"o}nwall's inequality yields (\ref{coupling_bound_a1}) and (\ref{coupling_bound_a2}) follows from another application of Lemma \ref{lem:a2}. 
\end{proof}

\begin{lem}\label{lem:a3}
Let $\Pi_1(t)$ and $\Pi_2(t)$ be $\mcF_t$-adapted unit rate Poisson processes and $t_1$ and $t_2$ be $\mcF_t$-stopping times with $t_1\vee t_2 \le D$ almost surely.  Then, 
\begin{equation*}
\expect[\Pi_1(t_1)^2+\Pi_2(t_2)^2]  \le (2D+1)\expect[t_1+t_2].
\end{equation*}
\end{lem}

\begin{proof}
From \cite[Lemma 3.1]{Chevallier2016} we have
\begin{equation*}
\expect[\Pi_j(t_j)^2]  = 2\expect[\Pi_j(t_j)t_j]  - \expect[t_j^2] + \expect[t_j], 
\end{equation*}
from which it follows that
\begin{align*}
&\expect[\Pi_1(t_1)^2+\Pi_2(t_2)^2] \\
&\quad\quad\quad= 2\expect[\Pi_1(t_1)t_1+\Pi_2(t_2)t_2]  - \expect[t_1^2+t_2^2] + \expect[t_1+t_2]\\
&\quad\quad\quad \le  2\expect[(t_1 \vee t_2)(\Pi_1(t_1)+\Pi_2(t_2))] + \expect[t_1+t_2]\\
&\quad\quad\quad \le 2D\expect[\Pi_1(t_1)+\Pi_2(t_2)] + \expect[t_1+t_2] = (2D+1)\expect[t_1+t_2].
\end{align*}
\end{proof}
\begin{proof}[Proof of Lemma \ref{lem:yfb}]

We first establish (\ref{coupling_bound1}):
\begin{align*}
||\Delta Y^{\Omega}(t)||^2 & \le_C \sum_{j \in \mcJ_2} (\Pi_{4,j}\left(A_{1,j}(t)\right) - \Pi_{5,j}\left(A_{2,j}(t) \right))^2\\
&\le_C \sum_{j \in \mcJ_2} (\Pi_{4,j}\left(A_{1,j}(t)\right) + \Pi_{5,j}\left(A_{2,j}(t) \right))^2\\
&\le_C \sum_{j \in \mcJ_2} \Pi_{4,j}\left(A_{1,j}(t)\right)^2 + \Pi_{5,j}\left(A_{2,j}(t) \right)^2
\end{align*}
In order to apply Lemma \ref{lem:a3} we need to bound the internal times. For $A_{1,j}(t)$, 
\begin{align*}
A_{1,j}(t) &=   \int_0^t  \alpha_j(X^{\Omega}(s),Y^X(s))- \alpha_j(X^{\Omega}(s),Y^X(s)) \wedge \alpha_j(Z(s),Y^Z(s)))ds\\
&\le   \int_0^t  \alpha_j(X^{\Omega}(s),Y^X(s))ds \le  t\alpha_j(K,c). 
\end{align*}
In fact, the same bound holds for $A_{2,j}(t)$. It follows that 
\begin{align*}
\expect[ \Pi_{4,j}\left(A_{1,j}(t)\right)^2 + \Pi_{5,j}\left(A_{2,j}(t) \right)^2] 
&\le (t\alpha_j(K,c)+1)\expect[ \Pi_{4,j}\left(A_{1,j}(t)\right) + \Pi_{5,j}\left(A_{2,j}(t) \right)]\\
&= (t\alpha_j(K,c)+1)\expect[ A_{1,j}(t) + A_{2,j}(t)]\\
&\le_C \int_0^t\expect[||X^{\Omega}(s)-Z(s)||]ds+\int_0^t\expect[||\Delta Y^{\Omega}(s)||]ds
\end{align*}
%

Theorem \ref{thm:a1} now implies (\ref{coupling_bound1}).
In order to obtain (\ref{coupling_bound2}), simply note $\Delta Y^{\Omega}(t)$ is bounded by a constant and apply (\ref{coupling_bound_a2}). 
\end{proof}
%

\end{appendix}

\bibliographystyle{elsarticle-num}

\bibliography{jcp2016-refs}{}

\begin{thebibliography}{10}
\expandafter\ifx\csname url\endcsname\relax
  \def\url#1{\texttt{#1}}\fi
\expandafter\ifx\csname urlprefix\endcsname\relax\def\urlprefix{URL }\fi
\expandafter\ifx\csname href\endcsname\relax
  \def\href#1#2{#2} \def\path#1{#1}\fi

\bibitem{bressloff2014book}
P.~C. Bressloff, {Stochastic Processes in Cell Biology}, Springer, 2014.

\bibitem{crudu2009}
A.~Crudu, A.~Debussche, O.~Radulescu, {Hybrid stochastic simplifications for
  multiscale gene networks}, BMC systems biology 3~(1) (2009) 1.

\bibitem{eve2007}
W.~E, D.~Liu, E.~Vanden-Eijnden, {Nested stochastic simulation algorithms for
  chemical kinetic systems with multiple time scales}, Journal of Computational
  Physics 221~(1) (2007) 158--180.
\newblock \href {http://dx.doi.org/10.1016/j.jcp.2006.06.019}
  {\path{doi:10.1016/j.jcp.2006.06.019}}.

\bibitem{jahnke2012}
T.~Jahnke, M.~Kreim, {Error bound for piecewise deterministic processes
  modeling stochastic reaction systems}, Multiscale Modeling {\&} Simulation
  10~(4) (2012) 1119--1147.

\bibitem{Chevallier2016}
A.~Chevallier, S.~Engblom, {Pathwise error bounds in Multiscale variable
  splitting methods for spatial stochastic kinetics} (2016).
\newblock \href {http://arxiv.org/abs/1607.00805} {\path{arXiv:1607.00805}}.

\bibitem{kurtz1972}
T.~G. Kurtz, {The relationship between stochastic and deterministic models for
  chemical reactions}, The Journal of Chemical Physics 57~(7) (1972)
  2976--2978.

\bibitem{Anderson2015book}
D.~F. Anderson, T.~G. Kurtz,
  \href{http://www.springer.com/us/book/9783319168944}{{Stochastic Analysis of
  Biochemical Systems}}, in: Stochastics in Biological Systems, Vol.~1,
  Springer International Publishing, 2015, p.~90.
\newblock \href {http://dx.doi.org/10.1007/978-3-319-16895-1}
  {\path{doi:10.1007/978-3-319-16895-1}}.
\newline\urlprefix\url{http://www.springer.com/us/book/9783319168944}

\bibitem{anderson2015b}
D.~F. Anderson, M.~Koyama, {Computational complexity analysis for Monte Carlo
  approximations of classically scaled population processes}, IMA Journal of
  Numerical Analysis 35~(4) (2015) 1757--1778.

\bibitem{Gillespie77}
D.~T. Gillespie, Exact stochastic simulation of coupled chemical reactions., J.
  Phys. Chem. 81 (1977) 2340--2361.

\bibitem{Gillespie07}
D.~T. Gillespie, Stochastic simulation of chemical kinetics, Ann. Rev. Phys.
  Chem. 58 (2007) 33--55.

\bibitem{Zeiser08}
S.~Zeiser, U.~Franz, O.~Wittich, V.~Liebscher, Simulation of genetic networks
  modelled by piecewise deterministic markov processes, IET Syst. Biol. 2
  (2008) 113--135.

\bibitem{Cao06}
Y.~Cao, D.~T. Gillespie, L.~R. Petzold, Efficient step size selection for the
  tau-leaping simulation method., J. Chem. Phys. 124 (2006) 044109.

\bibitem{anderson2011}
D.~F. Anderson, A.~Ganguly, T.~G. Kurtz, {Error analysis of tau-leap simulation
  methods}, The Annals of Applied Probability (2011) 2226--2262.

\bibitem{anderson2014}
D.~F. Anderson, D.~J. Higham, Y.~Sun, {Complexity of multilevel Monte Carlo
  tau-leaping}, SIAM Journal on Numerical Analysis 52~(6) (2014) 3106--3127.

\bibitem{anderson2012}
D.~F. Anderson, D.~J. Higham, {Multilevel Monte Carlo for continuous time
  Markov chains, with applications in biochemical kinetics}, Multiscale
  Modeling {\&} Simulation 10~(1) (2012) 146--179.

\bibitem{goodman2009}
J.~B. Goodman, K.~K. Lin, {Coupling control variates for Markov chain Monte
  Carlo}, Journal of Computational Physics 228~(19) (2009) 7127--7136.

\bibitem{giles2015}
M.~B. Giles, {Multilevel Monte Carlo methods}, Acta Numerica 24 (2015)
  259--328.

\bibitem{Ganguly2014}
A.~Ganguly, D.~Altintan, H.~Koeppl,
  \href{http://arxiv.org/abs/1409.4303}{{Jump-Diffusion Approximation of
  Stochastic Reaction Dynamics: Error bounds and Algorithms}}, arXiv (2014) 32.
\newline\urlprefix\url{http://arxiv.org/abs/1409.4303}

\bibitem{ethier86}
S.~N. Ethier, T.~G. Kurtz, {Markov processes : characterization and
  convergence}, Wiley, 1986.

\bibitem{sanft2015}
K.~R. Sanft, H.~G. Othmer,
  \href{http://scitation.aip.org/content/aip/journal/jcp/143/7/10.1063/1.4928635}{{Constant-complexity
  stochastic simulation algorithm with optimal binning}}, The Journal of
  Chemical Physics 143~(7) (2015) 074108.
\newblock \href {http://dx.doi.org/10.1063/1.4928635}
  {\path{doi:10.1063/1.4928635}}.
\newline\urlprefix\url{http://scitation.aip.org/content/aip/journal/jcp/143/7/10.1063/1.4928635}

\bibitem{Levien16}
E.~Levien, P.~C. Bressloff, A stochastic hybrid framework for obtaining
  statistics of many random walkers in a switching environment, Multiscale
  Modeling Simulation In press.

\bibitem{hufton2015}
P.~G. Hufton, Y.~T. Lin, T.~Galla, A.~J. McKane, {Intrinsic noise in systems
  with switching environments}, arXiv preprint arXiv:1512.00785.

\bibitem{Malrieu15}
F.~Malrieu, Some simple but challenging markov processes, Annales de la
  Facult{\'e} des Sciences de Toulouse. Math{\'e}matiques. S{\'e}rie 6 24
  (2015) 857--883.

\bibitem{gardiner1985}
C.~W. Gardiner, Others, {Handbook of stochastic methods}, Vol.~4, Springer
  Berlin, 1985.

\bibitem{duncan2016}
A.~Duncan, R.~Erban, K.~Zygalakis, {Hybrid framework for the simulation of
  stochastic chemical kinetics}.

\bibitem{Riedler2013}
M.~G. Riedler, \href{http://dx.doi.org/10.1016/j.cam.2012.09.021}{Almost sure
  convergence of numerical approximations for piecewise deterministic markov
  processes}, J. Comput. Appl. Math. 239 (2013) 50--71.
\newblock \href {http://dx.doi.org/10.1016/j.cam.2012.09.021}
  {\path{doi:10.1016/j.cam.2012.09.021}}.
\newline\urlprefix\url{http://dx.doi.org/10.1016/j.cam.2012.09.021}

\bibitem{veltz2015}
R.~Veltz, \href{http://arxiv.org/abs/1504.06873}{A new twist for the simulation
  of hybrid systems using the true jump method}, arXiv [math].
\newline\urlprefix\url{http://arxiv.org/abs/1504.06873}

\bibitem{Newby2012}
J.~M. Newby, {Isolating intrinsic noise sources in a stochastic genetic
  switch}, Physical Biology 9~(2) (2012) 026002.

\bibitem{Karmakar04}
R.~Karmakar, I.~Bose, Graded and binary responses in stochastic gene
  expression., Phys. Biol. 1~(197-204).

\bibitem{Thomas2014}
P.~Thomas, N.~Popovi{\'{c}}, R.~Grima,
  \href{http://www.pnas.org/content/111/19/6994$\backslash$nhttp://www.ncbi.nlm.nih.gov/pubmed/24782538}{{Phenotypic
  switching in gene regulatory networks}}, Proceedings of the National Academy
  of Sciences of the United States of America 111~(19) (2014) 6994--6999.
\newblock \href {http://dx.doi.org/10.1073/pnas.1400049111}
  {\path{doi:10.1073/pnas.1400049111}}.
\newline\urlprefix\url{http://www.pnas.org/content/111/19/6994$\backslash$nhttp://www.ncbi.nlm.nih.gov/pubmed/24782538}

\bibitem{bremaud2013}
P.~Br{\'{e}}maud, {Markov chains: Gibbs fields, Monte Carlo simulation, and
  queues}, Vol.~31, Springer Science {\&} Business Media, 2013.

\bibitem{anderson2015}
D.~F. Anderson, M.~Koyama, {An asymptotic relationship between coupling methods
  for stochastically modeled population processes}, IMA Journal of Numerical
  Analysis 35~(4) (2015) 1757--1778.

\bibitem{anderson2015ion}
D.~F. Anderson, B.~Ermentrout, P.~J. Thomas, {Stochastic representations of ion
  channel kinetics and exact stochastic simulation of neuronal dynamics},
  Journal of computational neuroscience 38~(1) (2015) 67--82.

\bibitem{Isaacson2009a}
S.~Isaacson, \href{http://epubs.siam.org/doi/abs/10.1137/070705039}{{The
  reaction-diffusion master equation as an asymptotic approximation of
  diffusion to a small target}}, SIAM Journal on Applied Mathematics 70~(1)
  (2009) 77--111.
\newblock \href {http://dx.doi.org/10.1137/070705039}
  {\path{doi:10.1137/070705039}}.
\newline\urlprefix\url{http://epubs.siam.org/doi/abs/10.1137/070705039}

\bibitem{Isaacson2009}
S.~A. Isaacson, D.~Isaacson, {Reaction-diffusion master equation,
  diffusion-limited reactions, and singular potentials}, Physical Review E -
  Statistical, Nonlinear, and Soft Matter Physics 80~(6).

\bibitem{hespanha2005}
J.~P. Hespanha, A.~Singh, {Stochastic models for chemically reacting systems
  using polynomial stochastic hybrid systems}, International Journal of robust
  and nonlinear control 15~(15) (2005) 669--689.

\bibitem{singh2007derivative}
A.~Singh, J.~P. Hespanha, {A derivative matching approach to moment closure for
  the stochastic logistic model}, Bulletin of mathematical biology 69~(6)
  (2007) 1909--1925.

\end{thebibliography}

\end{document}